 \documentclass[twocolumn]{autart}

\usepackage{xcolor}

\usepackage{changes}
\usepackage{cite}
\usepackage{amsmath,amssymb,amsfonts}
\usepackage{algorithm}
\usepackage{algpseudocode} 
\usepackage{tabularx}
\algtext*{EndWhile}
\algtext*{EndIf}
\algtext*{EndFor}

\usepackage{graphicx}
\usepackage{textcomp}
\usepackage{xcolor}

    \let\theoremstyle\relax
\usepackage{paralist} 
\usepackage{todonotes}
\usepackage[hyphens]{url}

\usepackage{changes}

\usepackage{tikz}
\usetikzlibrary{shapes,arrows,fit,backgrounds,automata}
\usetikzlibrary{shapes.geometric,fit}
\usepackage[utf8]{inputenc}

\tikzset{main node/.style={circle,draw,minimum size=1cm,inner sep=0pt},}
\tikzstyle{container} = [draw, rectangle, inner sep=0.1cm]
\tikzstyle{container-ellipse} = [draw, ellipse, inner sep=0.1cm]

\usepackage{acronym}
\usepackage{amsthm}

\usepackage{soul}
\usepackage{subcaption}

\definechangesauthor[name=Jie Fu, color=blue]{JF}

\definecolor{darkgreen}{rgb}{0,0.5,0}

\newcommand{\ie}{\textit{i.e.}}
 
\newcommand{\obs}{\mathsf{Obs}}
\newcommand{\dist}{\mathcal{D}}
 \newcommand{\sa}{\gamma}
\newcommand{\SA}{\Gamma}

\newcommand{\act}{\mathbf{A}}

\newcommand{\supp}{\mathsf{Supp}}
\newcommand{\plays}{\mathsf{Plays}}

\newcommand{\post}{\mathsf{Post}}
\newcommand{\prefplays}{\mathsf{Prefs}}

\newcommand{\asw}{\mathsf{ASW}}
\newcommand{\dasw}{\mathsf{DASW}}

\newcommand{\poswin}{\mathsf{PWin}}
\newcommand{\last}{\mathsf{last}}

\newcommand{\outcomes}{\mathsf{Outcomes}}

\newtheorem{theorem}{Theorem} 
\newtheorem{definition}{Definition}

\newtheorem{problem}{Problem}

\newtheorem{example}{Example}

\newtheorem{lemma}{Lemma}
\newtheorem{assumption}{Assumption}

\acrodef{pomdp}[POMDP]{partially observable Markov decision process}
\acrodef{mdp}[MDP]{Markov decision process}
\acrodef{asw}[ASW]{Almost-Sure Winning}
\acrodef{cps}[CPS]{Cyber-Physical Systems}
\acrodef{dasw}[DASW]{Deceptive Almost-Sure Winning}

\newcommand{\prog}{\mathsf{Prog}}

\newcommand{\hgame}{\mbox{HG}}

\newcommand{\detected}{\kappa_d}
\begin{document}
\begin{frontmatter}
\title{Reactive Synthesis of Sensor Revealing Strategies in  Hypergames on Graphs}

\thanks[footnoteinfo]{This paper was not presented at any IFAC 
meeting. Corresponding author S.~Udupa.}
\author[Gainesville]{Sumukha Udupa}\ead{sudupa@ufl.edu},    
\author[Adelphi]{Ahmed Hemida}\ead{ahmed.h.hemida.ctr@army.mil},
\author[Adelphi]{Charles A. Kamhoua}\ead{charles.a.kamhoua.civ@army.mil},
\author[Gainesville]{Jie Fu}\ead{fujie@ufl.edu}               

\address[Gainesville]{Department of Electrical and Computer Engineering, University of Florida, Gainesville, FL 32611, USA}

\address[Adelphi]{U.S. Army Research Laboratory, Adelphi, MD 20783, USA}


\begin{abstract}
In many security applications of cyber-physical systems, a system designer must guarantee that critical missions are satisfied against attacks in the sensors and actuators of the CPS. Traditional security design of CPSs often assume that attackers have complete knowledge of the system. In this article, we introduce a class of deception techniques and study how to leverage asymmetric information created by deception to strengthen CPS security.  Consider an adversarial interaction between a CPS defender and an attacker, who can perform sensor jamming attacks.  To mitigate such attacks, the defender introduces asymmetrical information by deploying a ``hidden sensor," whose presence is initially undisclosed but can be revealed if queried. 
We introduce hypergames on graphs to model this game with asymmetric information. Building on the solution concept called subjective rationalizable strategies in hypergames, we identify two stages in the game: An initial game stage where the defender commits to a strategy perceived rationalizable by the attacker until he deviates from the equilibrium in the attacker's perceptual game; Upon the deviation,  a delay-attack game stage starts where the defender plays against the attacker, who has a bounded delay in attacking the sensor being revealed.
Based on   backward induction, we develop
an algorithm that determines, for any given state, if the defender can benefit from hiding a sensor and revealing it later.  If the answer is affirmative,  the algorithm outputs a sensor revealing strategy to determine when to reveal the sensor during dynamic interactions. 
 We demonstrate the effectiveness of our deceptive strategies through two case studies related to CPS security applications.          
\end{abstract}
\end{frontmatter}

\section{Introduction}
With the integration of physical and advanced communication networks, intelligent and autonomous cyber-physical systems (CPSs) are increasingly assigned with mission-critical tasks in dynamic, uncertain environments \cite{dibaji2019systems, lin2017survey}. However, the interconnection between networks, sensors, and (semi-)autonomous systems introduces unprecedented vulnerabilities to both cyber and physical spaces \cite{dibaji2019systems,cardenas2008research,parkinson2017cyber,lun2019state}. 
Intentional attacks on a CPS sensor and actuator network can threaten the intricate attributes of the system \cite{cardenas2008research},  beyond traditional safety and stability. A notable example is the Iranian military's capture of a US RQ-170 drone through a sensor spoofing attack \cite{rawnsley2011iran}. A recent long-duration GPS jamming attack disrupted global satellite navigation systems for hundreds of aircraft in the Baltic region in Europe \cite{NewScientist}. Similarly, the increase in GPS spoofing and jamming incidents in the Mediterranean and Black Seas has impacted over a hundred cargo-carrying vessels \cite{LloydsList}.    



Researchers have thus extensively investigated the vulnerabilities in sensors, including spoofing attacks on LiDAR \cite{cao2019adversarial}, satellites \cite{rutkin2013spoofers},  and gyroscopic sensors  \cite{son2015rocking}.
Beyond spoofing, jamming has also emerged as a significant threat to sensor networks, due to its ease of execution and difficulty in detection \cite{mpitziopoulos2009survey,pirayesh2022jamming}. 
To strengthen the security of CPSs against such vulnerabilities, researchers have been developing a range of defense mechanisms, including game-theoretic approaches (see a recent survey \cite{tushar2023survey}), control theory approaches \cite{lun2019state,oliveira2023classification}, and learning-based methods \cite{wickramasinghe2018generalization}. 

In this work, we develop a game-theoretic approach to synthesize a provably secured CPS against sensor jamming attacks, leveraging the power of deception. Traditionally, control design for secured CPSs assume attackers possess complete knowledge of the system with possible imperfect observations, enabling a ``worst-case" scenario but limiting to conservative countermeasures. 
Deception hinges on relaxing the assumption of complete and symmetric information between the system defender and the attacker. Particularly, if an adversary has incomplete knowledge of the sensing or actuation capabilities, then there is an opportunity for the CPS controller to exploit the attacker's lack of information to achieve mission performance despite of attacks.  To this end, we model the CPS defender and attacker interactions using an incomplete and imperfect information game, where the system, referred to as player 1 or P1, must satisfy a reachability objective (\ie, reaching a set of goal states) with partial observations obtained from the sensor network. During the mission execution, the attacker, referred to as player 2 or P2, aims to prevent P1 from achieving his task by carrying out sensor attacks on a subset of vulnerable sensors. However, P2 lacks full knowledge of the system's sensing capabilities. We explore counterattack strategies with a class of capability deception: Could a strategic deployment of a ``hidden sensor" (a sensor whose presence is unknown to the adversary) offer a decisive advantage to the CPS in countering sensor attacks and ensuring reliable performance? If the hidden sensor would be revealed to the attacker on being queried by P1, when should  P1  query the hidden sensor to create the maximal strategic advantages?  

In literature, Bayesian games have been used to capture such interaction between players, when players have incomplete information \cite{harsanyi1967games}. In Bayesian games, each player has access to some private information that is encapsulated as a ``type". Each player defines a distribution over the possible types of the opponent, and these types are assumed to be public knowledge.  With this assumption,  one can transform a game with incomplete information into a game with imperfect information. However, in our context, we do not assume the attacker has constructed a finite set of possible hidden sensors deployed, preventing the modeling of Bayesian games.
Another commonly employed approach is to model the interaction between P1 and P2 using hypergames \cite{bennett1986hypergame, house2010hypergame}. A hypergame is a hierarchy of perceptual games, each of which models one player's subjective view of the interaction with others given his information and high-order information (what he knows that the other knows ...). 
We extend the hypergame model to games on graphs for analyzing the defender-attacker interactions with imperfect and incomplete information.  Since P1 knows P2's incomplete information, P1 can construct P2's perceptual game and use it to construct P1's own subjective game including the true interaction dynamics and P2's misperceived interaction dynamics. The hierarchy of information leads to a hierarchy in perceptual games.

In particular, the hypergame evolves with information leakage during their interaction: 1) If P1 deviates from the equilibrium in P2's perceptual game without using a hidden sensor, then P2 would infer that P2's game model is incorrect but does not know the true game; 2) If P1 deviates from the equilibrium by revealing the hidden sensor, then P2 would not only infer that P2's game model is incorrect but also know the true game. Given a qualitative objective for P1 is to achieve the task with probability one, our analysis shows that there is no benefit for P1 to deviate without revealing the hidden sensor. Further, to synthesize a deceptive sensor revealing strategy that capitalizes on P2's incomplete knowledge,  we decompose the interaction between P1 and P2 into two parts: the \emph{initial game} and the \emph{delay-attack game}, as shown in Figure \ref{fig:hidden_sensor_conceptual_diagram}. The initial game captures the interaction before the revelation of the hidden sensor, while the delay-attack game represents the interaction afterward. We employ backward induction, where the solution of the delay-attack game is used to define the objective for the initial game. And the solution of the initial game determines P1's sensor-revealing strategy.

\begin{figure}
    \centering
    \includegraphics[width=0.9\linewidth]{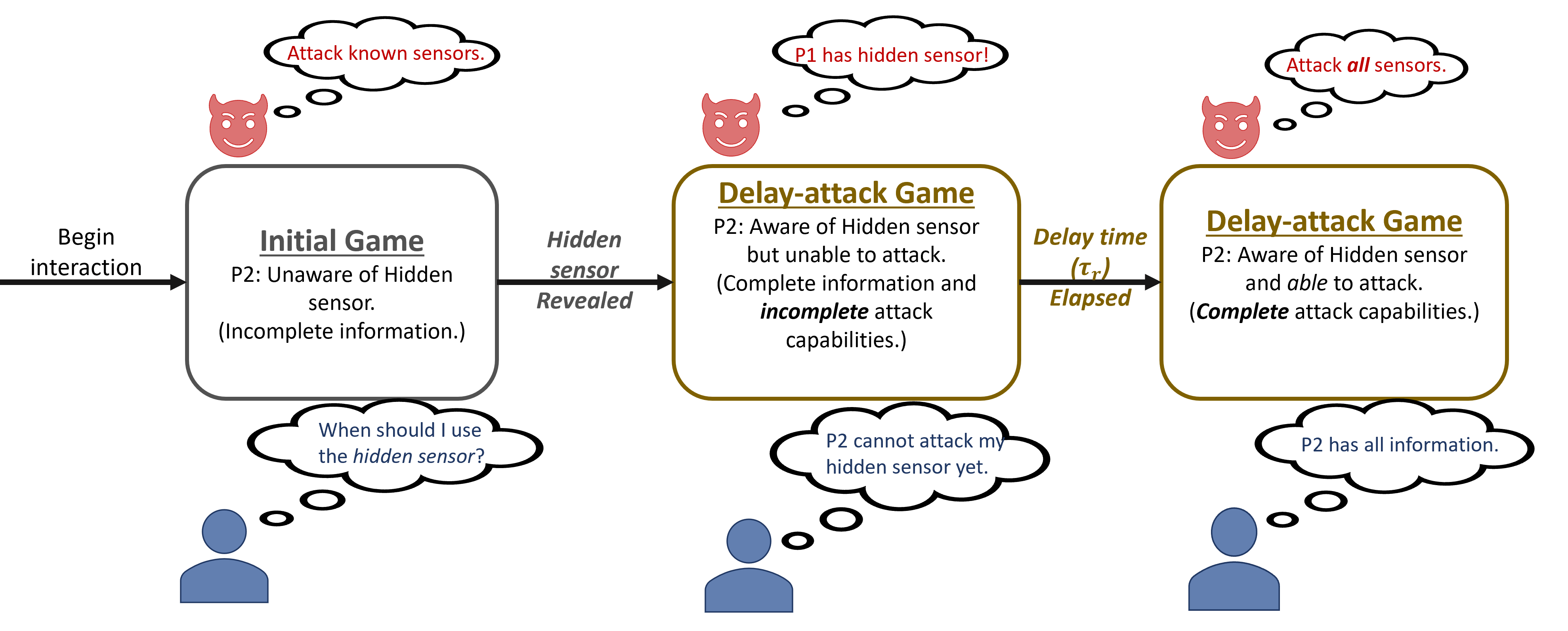}
    \caption{Interaction between P1 (indicated in blue) and P2 (indicated in red) with P1 having a hidden sensor.}
    \label{fig:hidden_sensor_conceptual_diagram}
\end{figure}



Our work is closely related to the problem of supervisory control for discrete event systems (DES), particularly in the context of sensor deception and actuator attacks. There has been extensive research on sensor attacks in the literature on supervisory control of discrete event systems \cite{hadjicostis2022cybersecurity,oliveira2023classification}.
In the domain of sensor deception using supervisory control, two main directions of research have been investigated. One is to design stealthy sensor deception attacks from the attacker's perspective \cite{meira2020synthesis, 8264281, 8814740, 9477153, lin2020synthesis, tai2021synthesis} and the other is to design a sensor deception attack resilient supervisor from the defender/system's perspective \cite{lin_bounded_resilient_sup_2019, romulo_resilient_sup_2019, romulo_robust_sup_2021, su2022existence, wang2019supervisory}. Primarily, in deterministic and stochastic DES literature, sensor deception is defined as any attack on the sensor/sensor information that the supervisor receives and thus is modeled as editing of observation strings. As mentioned earlier, a general assumption in the literature of sensor deception is that the attacker has complete knowledge of the plant and the supervisor, and undertakes attacking actions to deceive the supervisor \cite{oliveira2023classification}. The attack-resilient supervisor is synthesized by considering an augmented plant that captures the interaction with the attacker. The assumption of the attacker having complete knowledge could be considered overly restrictive and it may not always align well with the real-world situations \cite{oliveira2023classification}.


In the literature on deception, studies such as \cite{karabag2021deception, kulkarni2019signaling, li2022dynamic} focus on intention/payoff deception where the two players have different perceptions of the opponent's intention.  Capability deception has been studied for both two-player turn-based deterministic games \cite{kulkarni2021synthesisdeception} and concurrent stochastic games \cite{shi2023quantitative}.  Both works consider the setting where P1 has a hidden action and decides when to employ the hidden action to exploit P2's initial incomplete knowledge about the game. They considered games with perfect observations instead of partial observations. In our setting of control against sensor attacks, the defender, P1, has partial observations which are partially controlled by P2, who carries out jamming attacks. The solution approach is thus different from the aforementioned capability deception with perfect observations.  This work is built on our previous work \cite{lcss_paper} on the reactive synthesis of attack-aware controllers with active perception and control. Instead of assuming both players have complete information about the game dynamics, we introduce asymmetric information and develop joint active perception and control that leverages capability deception.


Our contributions in this work are concluded as follows. \begin{enumerate}
    \item \textbf{Hypergame modeling  deception with a hidden sensor:} We begin by formulating the interaction between P1 and P2 as a zero-sum stochastic reachability hypergame on graphs, where
   P1's observations are jointly controlled by active sensing and  P2's sensor attacks, creating a dynamic observation function. 
    \item \textbf{Algorithm for synthesizing deceptive strategies:}  We introduce a delay in P2's reaction to the revelation of the hidden sensor and utilize the hypergame structure to build an augmented game, expanding the planning state space to include the true state, P1's belief, and the delay in P2's reaction upon the revelation of the hidden sensor. Leveraging the concept of subjectively rationalizable strategies \cite{subjective_rationalizability_paper}, we model the attacker's response given evolving knowledge about P1's sensing capability. Through backward induction, we decide for each possible initial state, whether there exists a deceptive strategy for P1 that guarantees to satisfy the reachability objective with probability one.
    \item \textbf{Theoretical characterization of deception advantages with hidden sensor:} We delve into theoretical analysis of P1's strategic advantage with deception and prove that for qualitative objective (winning almost surely), there is no advantage to deviate without revealing the hidden sensor. We formally define the value of deception to quantify P1's strategic advantage.
    \item \textbf{Experimental demonstrations:}  We analyze the value of deception in revealing the hidden sensor at each state and demonstrate the effectiveness of the developed methods in the simulation of two different experimental scenarios.
    
\end{enumerate}

The remainder of the paper is organized as follows. 
Section \ref{sec:preliminaries_and_problem_formulation} introduces the necessary preliminaries and problem formulation along with the introduction of a running example. Section \ref{sec:Main_results} discusses the main results of the paper with the construction of the augmented games extending hypergame model, main theoretical results, and synthesis of deception strategy. In Section \ref{sec:Experimental_case_studies} we demonstrate the developed theory in two different scenarios. Section \ref{sec:Conclusion} concludes and discusses future work on leveraging deception for CPS security.

\section{Preliminaries and Problem Formulation}
\label{sec:preliminaries_and_problem_formulation}

\paragraph*{Notations} We use the notation $\mathcal{A}$ for a finite set of symbols, also known as alphabet. A sequence of symbols $w = w_0 w_1 \cdots w_n$ with $w_i \in \mathcal{A}$ for any $0 \leq i \leq n$ is called a finite word, and $\mathcal{A}^*$ is the set of all finite words that can be generated with the alphabet $\mathcal{A}$. The empty word or the empty sequence is denoted by $\epsilon$. We denote the set of all $\omega$-regular words as $\mathcal{A}^{\omega}$ obtained by concatenating the elements in $\mathcal{A}$ infinitely many times. The length of a word is given by $|w|$. For a word $w=uv$, $u$ is a prefix of $v$ for $u \in \mathcal{A}^{*}$ and $v \in \mathcal{A}^{\omega}$. Given a finite set $X$, let $\dist(X)$ be the set of all probability distributions over $X$. Given a distribution $d\in \dist (X)$, let $\supp(d) = \{x\in X\mid d(x)>0\} $ be the support of this distribution. We also have the powerset of $X$ as $2^X$. 

We introduce the class of two-player stochastic reachability game with active sensing and sensor attacks. 
In this game, an autonomous agent (Player 1/defender), actively queries sensors to obtain task-relevant information. 
Meanwhile, an attacker (Player 2), seeks to compromise P1's mission by reactively attacking information requested by P1.
\\
\begin{definition}[Stochastic Reachability Game with Active Sensing and Sensor Attacks]
	\label{def:cps-game}
The \emph{two-player stochastic game with active sensing and sensor attacks} is a tuple
\[
G = (S, A_1 ,\Sigma, \SA_1, \SA_2, \mathbf{P}, O, \obs, s_0, o_0, F)
\]
in which

\begin{itemize}
\item $S$ is a finite set of states.
 \item $A_1$ is the set of P1's \emph{control actions}.
 \item $\Sigma = \{\sigma_0,\sigma_1,\cdots,\sigma_N\}$ is a set of \emph{indexed sensors}.
 \item Each sensor query action for P1 refers to a subset of sensors, $\SA_1 \subseteq 2^{\Sigma}$, queried and each sensor attack action for P2 refers to subset of sensors, $\SA_2 \subseteq 2^{\Sigma}$, being attacked.
 \item $\mathbf{P}: S\times A_1 \rightarrow \dist(S)$ is a \emph{probability transition function} such that for each $s, s' \in S$ and $a \in A_1$, $\mathbf{P}(s,a,s')$ is the probability of reaching $s'$ given action $a$ taken at state $s$.
 \item $O \subseteq 2^S$ is the set of all \emph{observations}.
 \item $\obs: S\times \SA_1\times \SA_2\rightarrow O$ is the \emph{ observation function} such that $\obs(s,\sa_1,\sa_2)\in O$ is the observation P1 perceives when the current state is $s$, P1  performs sensing action $\sa_1$ and P2 takes a sensor attack action $\sa_2$. In this game, P2 is assumed to have a perfect observation of states and P1's actions, thereby P2's observation function is omitted. 
 \item $s_0 \in S$ is the initial state.
 \item $o_0 \in O$ is the initial observation and $s_0 \in o_0$.
 \item $F \subseteq S$ is P1's \emph{reachability objective}---a set of \emph{goal} states that P1 intends to reach.
\end{itemize}
\end{definition}

P1's observation function is defined as follows: At each state $s\in S$, with queried sensors by P1, $\gamma_1 \in \SA_1$, and attacked sensors by P2, $\gamma_2 \in \SA_2$, P1 receives the sensor readings from the sensors in $\gamma_1\setminus \gamma_2$ about the state $s$, which is observation $\obs(s, \gamma_1,\gamma_2)$. We assume that the observation function is deterministic. Two states $s, s' \in S$ are said to be \emph{observation equivalent given the sensing action $\gamma_1$  and sensor attack action $ \gamma_2$} if $\obs(s,\gamma_1,\gamma_2)= \obs(s',\gamma_1,\gamma_2)$.



In contrast to the standard POSG models \cite{chatterjee2007algorithms} where the observation functions are fixed, in our game, the observation function of P1 is determined dynamically during every round by P1's active sensing and P2's sensor attacks. 
    %
    %
    

%
%
%

\paragraph*{Game Play} The game play in $G$ is constructed as follows. The initial state $s_0$ and the initial observation of P1 is $o_0$. 
Based on the observation, P1 selects a control action $a_0 \in A_1$ and a sensor query action $\gamma_1 \in \SA_1$. The system moves to state $s_1$ with probability $P(s_0, a_0, s_1)$. P2 perfectly observes the state $s_1$ and P1's sensor query action $\gamma_1$,  selects an attack action $\gamma_2 \in \SA_2$. The system generates a new observation $o_1 = \obs(s_1,\gamma_1,\gamma_2)$ for P1, and then P1 selects another control action, and then the game continues.  
We denote the resulting play as $\rho = s_0 (a_0, \gamma_{1}^0)  s_1  \gamma_{2}^0 (a_1, \gamma_{1}^1)  s_2 \gamma_{2}^1 \ldots$. 
The set of plays in $G$ is denoted by $\plays(G)$. A prefix of a play $\rho$ is a sub-sequence of states and actions $\upsilon = s_0 (a_0, \sa_1^0) s_1 \cdots s_k$, $k \ge 0$,   and the set of prefixes of plays in $G$ is denoted by $\prefplays(G)$. An infinite play $\rho   = s_0 (a_0, \gamma_{1}^0)  s_1  \gamma_{2}^0 (a_1, \gamma_{1}^1)  s_2 \gamma_{2}^1 \ldots$ is \emph{winning} for P1 if $s_k\in F$ for some $k\ge 0$. Otherwise, it is winning for P2.

\paragraph*{Observation Equivalent Plays to P1} Given a play $\rho   = s_0 (a_0, \gamma_{1}^0)  s_1  \gamma_{2}^0 (a_1, \gamma_{1}^1)  s_2 \gamma_{2}^1 \ldots$,  P1's \emph{observation} of  $\rho$ is $ \rho^o= o_0 (a_0, \gamma_{1}^0,\gamma_{2}^0) o_1 (a_1, \gamma_{1}^1, \gamma_{2}^1)\ldots $ where $o_{i+1}=\obs(s_{i+1}, \gamma_{1}^i, \gamma_{2}^i)$ for all $i\ge 0$ and $o_0$ is the initial observation. For notation convenience, we denote the observation of play $\rho$ as $\obs(\rho)$.
Two plays (or play prefixes) $\rho, \rho'$ are said to be observation-equivalent to P1, denoted by $\rho \sim \rho'$, if and only if $\obs(\rho)=\obs(\rho')$. 

\paragraph*{Players' Strategies}   
We denote P1's set of perception-control actions by $\act_1  = A_1 \times \SA_1$, and that of P2 by $\act_2= \SA_2$.
 A finite-memory, randomized (resp., deterministic) strategy for player $j \in \{1, 2\}$ is a function $\pi_j : \prefplays(G) \rightarrow \dist(\act_j)$ (resp., $\pi_j : \prefplays(G) \rightarrow \act_j$). 
 Player $j$ is said to follow strategy $\pi_j$ if for any prefix $\rho\in \prefplays(G)$ at which $\pi_j$ is defined, player $j$ takes the action $\pi_j(\rho)$ if $\pi_j$ is deterministic, or an action $a \in \supp(\pi_j(\rho))$ with probability  $\pi_j(\rho,a)$ if $\pi_j$ is randomized. $\Pi_j$ denotes the set of strategies of player $j$. 
A strategy is said to be \emph{observation-based} if $\pi_j(\rho) = \pi_j(\rho')$ whenever $\rho \sim \rho'$. P1's observation-based strategies are denoted by $\Pi_1^o$. A pair $(\pi_1,\pi_2)$ of players' strategies is called a \emph{strategy profile}.

The set of all possible plays starting from a state $s\in S$, with P1's strategy $\pi_1$ and P2's strategy $\pi_2$ in the game G is denoted by $\mathsf{Outcomes}_G(s,\pi_1,\pi_2)$. Formally, $\mathsf{Outcomes}_G(s,\pi_1,\pi_2) = \{\rho\in \plays(G) \mid \Pr(\rho; G^{\pi_1,\pi_2})>0\}$ where $\Pr(\rho; G^{\pi_1,\pi_2})$ is the probability of a play $\rho$ in the stochastic processes $G^{\pi_1,\pi_2}$ induced by the strategy profile $(\pi_1,\pi_2)$.

 \paragraph*{Almost-Sure Winning Strategy/Region} Given a reachability objective $F$, a strategy $\pi_1 \in \Pi_1$ is said to be \emph{almost-sure winning} for P1 if, for any strategy of P2, $\pi_2 \in \Pi_2$, P1 is guaranteed to visit $F$ with probability one. A strategy $\pi_2 \in \Pi_2$ is said to be \emph{almost-sure winning} for P2  if, for any strategy of P1, $\pi_1 \in \Pi_1$, P1 is guaranteed to visit $F$ with probability zero. 
 A set of prefixes of plays from which player $i$ has an almost-sure winning strategy is called player $i$'s \emph{almost-sure winning region},  denoted by $\asw_i(G)$ for $i \in \{1,2\}$. The almost-sure winning region for P1 can be obtained using Algorithm \ref{alg:posg-reachability} \cite{lcss_paper}, given in the Appendix \ref{appendix:almost_sure_winning_algo}.

\paragraph*{Positive Winning Strategy/Region} A strategy $\pi_1 \in \Pi_1$ is said to be \emph{positive winning} for P1 over a reachability objective to visit $F$ if, for any strategy of P2, $\pi_2 \in \Pi_2$, P1 has a positive probability to reach $F$. A set of prefixes of plays from which P1 has a positive winning strategy is called P1's positive winning region. The positive winning strategy/region for P2 is defined analogously. 

\subsection{Problem formulation of sensor-revealing game}
 
Under the threat of P2's sensor attack, P1 may benefit from having sensors unknown to P2, resulting in asymmetric information between players. Specifically, the following information structure is considered:
\begin{assumption}
\label{assumption:hidden_sensor}
Information structure:
\begin{itemize}
    \item 
P1 has a hidden sensor $\sigma_0 \in \Sigma$, \ie, this sensor is unknown to P2 initially. 
\item P2 has perfect observations of game play. That is, P2 observes states, P1's control and sensing actions, and also knows her attack actions.
\end{itemize}
\end{assumption}

Informally, for a given prefix of a game play, we say that P1 has \emph{qualitative strategic advantages} with a hidden sensor if the following conditions are satisfied: 1) without the hidden sensor, P1 cannot achieve his objective with probability one; 2) when P2 knows all sensors in the beginning, P1 cannot achieve his objective with probability one; 3) with the hidden sensor unknown to P2 initially, P1 can achieve his objective with probability one, by strategically deciding when to reveal his hidden sensor.

We introduce delay in P2's response after P1 reveals his hidden sensor to capture the three cases: 1) P2 can attack any sensor immediately when it is public knowledge (zero delay); 2) P2 can attack a revealed sensor after some delay (positive delay); 3) P2 cannot attack a revealed sensor (delay equals to infinite). This delay is motivated by the need to develop a technical approach and by practical applications, 
for instance, P2 may take some time to physically install jamming/attacking systems near the newly revealed sensor.

\begin{definition}[Delay in reaction]
    Let $t_1$ be the time step when P2 detects the hidden sensor, and $t_2$ be the time step when P2 can initiate attacks on the hidden sensor. A \emph{delay in reaction} is $\tau_r = t_2 - t_1$. 
\end{definition}

Trivially, we have that $\tau_r\ge 0$. We also assume that the delay is bounded. Otherwise, P2 knows the sensor but can never attack it. There is no need for P1 to hide this sensor.

The following assumption is made:
\\
\begin{assumption}
Delay in reaction:
 \begin{itemize}
     \item P2's delay in reaction is upper bounded by a finite maximal delay $\tau_r \le \overline{\tau}_r<\infty$.
  \item 
P1 knows P2's delay in reaction ($\tau_r$).
  \end{itemize}
\end{assumption}

Thus, we pose the problem of strategically revealing the hidden sensor. 
\\
\begin{problem}
Given the game $G$ in Def. \ref{def:cps-game}, a sensor $\sigma_0$ hidden from P2, and a delay in reaction $\tau_r \ge 0$, for a given $\rho \in \prefplays(G)$ from which P1 has no almost-sure winning, observation-based strategy when P2 has complete information of P1's sensors, determine if P1  has an observation-based strategy to ensure satisfying the reachability objective with probability one if P1  has sensor $\sigma_0$ hidden from the beginning and reveals the hidden sensor strategically. 
\end{problem}

We illustrate our problem setup and solution approaches using a running example. 
\\
\begin{example}  
\label{example:running_example}
Consider a simple CPS game depicted in Fig.~\ref{fig:running_example_1}. The game states are $S= \{s_i, i=1,2,3,4,5\}$
and P1 has two actions $A_1 = \{a,b\}$.
For clarity, the exact transition probabilities are omitted. From state $s_1$, with either action $a$ or $b$, P1 can reach either state $s_2$ or $s_3$ with a positive probability.
 The set of sensors deployed is $\{A, B, C, D\}$, which are range sensors. The sensor $A$ covers the states $\{s_2, s_3\}$, sensor $B$ covers the state $\{s_3\}$, sensor $C$ covers the states $\{s_4, s_5\}$ and the sensor $D$ covers the states $\{s_2, s_3, s_4\}$. This sensor coverage is shown in different colors in  Fig.~\ref{fig:running_example_1}. If the current state is in the sensor coverage, the sensor outputs $1$, $0$ otherwise. P1 is allowed to query any two of the sensors at a time and P2 is allowed to attack any one of the sensors.


 Among the sensors, sensor B is hidden and unknown to P2 at the start of the game. P1 aims to reach the final winning state $s_5$. The state $s_4$ is a sink state. In this example, when P1 does not query the hidden sensor, the states $\{s_2, s_3, s_5\}$ are in P1's almost-sure winning region. And $\{s_1\}$ is in P1's positive winning region. This is because, from state $s_1$, no matter whether P1 takes action $a$ or $b$, P1 cannot distinguish whether the reached state is $s_2$ or $s_3$ and thus cannot ensure to reach $s_5$ with probability one. If the actual state is $s_3$, choosing action $b$ would lead to the sink state $s_4$, and if the actual state is $s_2$, choosing action $a$ would lead to the sink state.
  
When P1 is allowed to query the hidden sensor $B$ and P2 is aware of the sensor $B$, it turns out that P1's almost-sure winning region is the same set $\{s_2,s_3,s_5\}$. This is because P2 can attack the sensor to reach the belief state $\{s_2,s_3\}$ from which P1 has no almost-sure winning actions.
We will later use this example to construct P1's deceptive, sensor reveal strategy and show the advantages of a hidden sensor $B$.
\end{example}

\begin{figure}[ht!]
        \centering       
       \begin{tikzpicture}[->,>=stealth',shorten >=1pt,auto,node distance=2.5cm,scale=0.7,semithick, transform shape]
    \tikzstyle{every state}=[fill=black!10!white];


    \node[state, initial] (1) at (-3, -0.5) {$s_1$};
    \node[state] (2) at (0, 1) {$s_2$};
    \node[state] (3) at (0, -1.5) {$s_3$};
    \node[state] (4) at (3, 1) {$s_4$};
    \node[state, accepting] (5) at (3, -1.5) {$s_5$};

    \node [container,fit=(2) (3),draw=red,dashed,line width=0.2mm ] (container) {};
    \node [container,fit=(4) (5),draw=green,dashed,line width=0.2mm ] (container) {};
    \draw [rotate=0,draw=blue, dashed,line width=0.5mm] (0, -1.5) ellipse (0.60cm and 0.60cm);
    \draw [draw=violet, line width =0.4mm, dashed] 
(-0.5, 1.5) 
-- (4, 1.5)
-- (-0.5, -3) 
-- cycle;

      \path 
        (1) edge   node[pos=0.3]{$a, b$} (2)
        (1) edge   node[pos=0.3]{$a, b$} (3)
        (2) edge   node[pos=0.3]{$a$} (4)
        (2) edge   node[pos=0.3]{$b$} (5)
        (3) edge   node[pos=0.3]{$b$} (4)
        (3) edge   node[pos=0.3]{$a$} (5)
        (4) edge[loop above]   node{$a, b$} (4)
        (5) edge[loop below]   node{$a, b$} (5)
        ;

\end{tikzpicture}
        \caption{An illustrative running example. The dashed lines represent the sensors: A (red), B (blue), C (green) and D (purple).
        }
      \label{fig:running_example_1}
\end{figure}
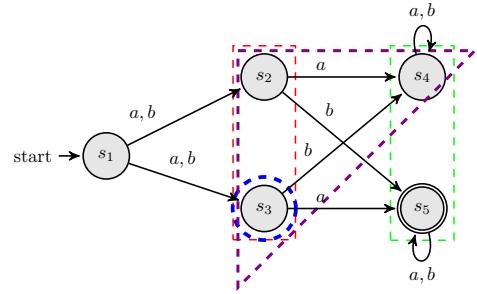

\section{Main results}
\label{sec:Main_results}

 In this section, we first extend the hypergame model \cite{bennett1980hypergames, house2010hypergame} to capture the asymmetric information between players. Then, based on the solution concept of hypergames, we define the construction of a sensor-revealing game, whose solution provides an almost-sure winning strategy for P1 using sensor deception, against the best response strategy of P2 given P2's subjective view of the interaction. 


\subsection{Hypergame Model}
\begin{definition}[Hypergame \cite{bennett1980hypergames}]
A level-1 two-player hypergame is a pair $\mathcal{HG}^1 = \langle G_1, G_2 \rangle$, where $G_1$ and $G_2$ are games perceived by players $P1$ and $P2$, respectively. A level-2 two-player hypergame is a pair $\mathcal{HG}^2 =  \langle \mathcal{HG}^1, G_2 \rangle$, where P1 perceives the interaction as a level-1 hypergame and P2 perceives it as game $G_2$. The first component of the hypergame is called the perceptual game of P1 and the second component is called the perceptual game of P2.  
\end{definition}

While it is possible to have a level-$k$ hypergame, it is sufficient to have a level-2 hypergame to model the interaction in the current problem setting.

$\\$
\begin{definition}[P2's perceptual game]
\label{def:P2-game}
Given the stochastic game $G = (S, A_1 ,\Sigma, \SA_1, \SA_2, \mathbf{P}, O, \obs,\\ s_0, o_0, F) $ and that sensor $\sigma_0\in \Sigma$ is hidden from P2, 
P2's perceptual game is a tuple 
\[
G^2 = (S, A_1, \Sigma \setminus \{\sigma_0\},
\SA_1^2 \cup \SA_2^2, P, O, \obs^2, s_0, o_0, F)
\]
where
\begin{itemize}
    \item P1's sensor query actions known to P2 is $\SA_1^2 = \{\sa_1 \mid \sa_1 \in \SA_1$ and $\sigma_0 \notin \sa_1\}$ and P2's sensor attack actions known to P2 is $\SA_2^2 = \{\sa_2 \mid \sa_2 \in \SA_2$ and $\sigma_0 \notin \sa_2\}$,
    \item $\obs^2: S\times \SA_1^2\times \SA_2^2\rightarrow O$ is P1's observation function known to P2.
\end{itemize}
The other components are the same as in the game $G$.
\end{definition}

Given P2's perceptual game known to P1, the interaction between two players is characterized by a level-2 hypergame. Next, 
We extend the hypergames from normal-form games \cite{bennett1986hypergame, bennett1980hypergames,vane2006advances} to the two-player stochastic game with a hidden sensor.
\begin{definition}\label{def:static_hypergame_1}
P1's perceptual game is a level-1 hypergame  defined as \[{\hgame}^1 = \langle G, G^2 \rangle,\] where $G$ is P1's knowledge about the actual game (see Def.~\ref{def:cps-game}) and $G^2$ is P2's perceptual game (see Def.~\ref{def:P2-game}). The interaction between P1 and P2 is a level-2 hypergame: \[
    {\hgame}^2  = \langle {\hgame}^1, G^2 \rangle.
  \] 
 \end{definition}
    

 \subsection{Synthesis of deceptive almost-sure winning strategies}

 To investigate how P1 can leverage the effect of deception with a hidden sensor, and P2's delay in reaction to achieve P1's objective, we construct a 
 \emph{sensor-revealing game}, a two-player stochastic game whose states are augmented with P1's and P2's beliefs, capabilities, and knowledge. 

 To facilitate the definitions, we introduce the function $\post_G:2^S \times A_1 \rightarrow 2^S$ that maps a set of states $B \subseteq S$ and an action $a \in A_1$ to the possible reachable states with the action $a$, $\post_G(B, a) = \{s' \in S \mid \exists s \in B : \mathbf{P}(s, a,s') > 0\}$. In the rest of the paper, we denote $\post_G(\{s\}, a)$ as $\post_G(s, a)$.
\\
 
 \begin{definition}[Sensor-Revealing Game]
 \label{def:sen_rev_game}
Given the stochastic game $G$ and P2's perceptual game $G^2$, with the delay in reaction $\tau_r$, the stochastic sensor-revealing game is a tuple

    	\[
	\mathcal{H} = \langle Q \cup \{q_F\}, (A_1 \times \SA_1) \cup \SA_2, \delta,\obs,\tau_r, q_0  \rangle,
	\]
where,
\begin{itemize}

    \item $Q = Q_1 \cup Q_N \cup Q_2 $, is the set of states consisting of P1, nature's states and P2's states. $Q_1 = \{(s,B,\detected) \mid s \in S, B \subseteq S, \detected \in \{-1,0,1,\cdots, \tau_r\} \}$ is the set of states where P1 selects a (control and sensing) action $(a, \sa_1)$. $Q_N = \{(s,B,a,\sa_1,\detected) \mid s \in S, B \in 2^S, (a,\sa_1) \in \act_1, \detected \in \{-1,0,1,\cdots, \tau_r\} \}$ is the set of nature's state. $Q_2 = \{(s,B,\sa_1,\detected) \mid s \in S, B \in 2^S, \sa_1 \in \SA_1, \detected \in \{-1,0,1,\cdots, \tau_r\} \}$ is the set of states where P2 selects a sensor attack action. 
   

    \item $q_F$ is the fictitious final state. It is a sink state. 
    
    \item $\act_1 =A_1\times \SA_1$ is a set of P1's actions and $\act_2= \SA_2$ is a set of P2's actions.
    \item $q_0 = (s_0,B_0,\detected=-1)$ is the initial state.

    \item $\delta: (Q_1 \times \act_1) \cup Q_N \cup (Q_2 \times \act_2) \rightarrow \dist(Q \cup \{q_F\})$ is the probabilistic transition function  defined as follows:
    \begin{enumerate}
    %
			\item For a P1's  state $(s, B, \detected) \in Q_1$, there are following cases: \begin{enumerate}[1)]
			    \item When an action $(a, \sa_1) \in \act_1^2$ and $\detected = -1$, \ie, an action from the set of P1's actions known to P2 is taken and the hidden sensor has not been detected,   $\\ \delta((s,B,\detected),(a, \sa_1),(s,B',a,\sa_1,\detected')) = 1$, where $B'= \post_G(B,a)$ and $\detected' =\detected $.
			    
			    \item When P1 chooses to query the hidden sensor, $(a, \sa_1) \in \act_1\setminus \act_1^2$ or $-1 < \detected < \tau_r$, \ie,  P2 knows the hidden sensor and cannot attack it,\\ $\delta((s,B,\detected),(a, \sa_1),(s,B',a,\sa_1,\detected')) = 1$.
			    Where $B'= \post_G(B,a)$ and  $\detected' = \detected+1$.
			    \item When $\detected = \tau_r$ \ie,  P2 has the ability to attack the hidden sensor, P1 selects any action $(a,\sa_1) \in \act_1$,  $\delta((s,B,\detected),(a, \sa_1),(s,B',a,\\\sa_1,\detected')) = 1.$
			    Where $B'= \post_G(B,a)$ and    $\detected'=\detected$.
			\end{enumerate}  
			\item For a nature's state $(s,B',a,\sa_1,\detected) \in Q_N$, we distinguish three cases:

   \begin{enumerate}[a)]
    		\item If $\post_G(s,a)\subseteq F$ then $\delta((s, B',a,\sa_1,\detected),\\q_F)=1$.  
    		\item If 	$\post_G(s,a)\cap F =\emptyset$, then   $\delta((s,B',a,\sa_1,\detected),\\ (s',B',\sa_1,\detected)) =P(s,a,s')$. 
    		\item If $\post_G(s,a)\cap F\ne \emptyset$ and $\post_G(s,a) \setminus F \ne\emptyset$, then, $\delta((s,B',a,\sa_1,\detected),q_F) =\epsilon$ and $\delta((s,B',a,\sa_1,\detected), (s',B',\sa_1, \detected)) =(1-\epsilon)\cdot P(s,a,s')$ where   $\epsilon \in (0,1)$ is a small constant. 
    			 That is, with some positive probability $\epsilon$, the final state $q_F$ is reached.  
			 \end{enumerate}

			\item For a P2's state $(s',B',\sa_1,\detected) \in Q_2$, we have the following two cases:
			\begin{enumerate}[1)]
			    
			    
			    \item When $-1 \le \detected < \tau_r$, 
       P2 chooses a sensor attack action $\sa_2 \in \act_2^2$,  $$\delta((s',B',\sa_1,\detected),\sa_2,(s', B'',\detected')) = 1$$ where $B'' =B'\cap \obs(s', \sa_1,\sa_2)$ and $\detected' = \detected$.

			    \item When $\detected \geq \tau_r$, \ie,  the hidden sensor is detected and P2 is capable of attacking it, P2 chooses an attack action  $\sa_2 \in \act_2$,  $\delta((s',B',\sa_1,\detected),\sa_2,(s', B'',\detected')) = 1$ where $B'' =B'\cap \obs(s', \sa_1,\sa_2)$ and 
			    $\detected'= \detected$.
			  
			\end{enumerate}
		\end{enumerate}
\end{itemize}

 \end{definition}
 
 A sequence of transitions $(s,B,\detected) \xrightarrow{(a,\sa_1)} (s, B', a,\sa_1,\\\detected') \dashrightarrow (s',B', \sa_1, \detected')\xrightarrow{\sa_2} (s',B'',\detected')$ is understood as follows: At the P1 state $(s,B,\detected)$, the true state is $s$ and P1 believes that any state in $B$ can be the true state and different values of $\detected$ represents if the hidden sensor has been revealed to P2 ($\detected> -1$) or not ($\detected = -1$). If it is revealed($\detected> -1$), then  if P2 is capable of attacking the hidden sensor ($  \detected \geq \tau_r$) or not ($0\le \detected < \tau_r$). P1 selects a pair of control and sensing actions $(a,\sa_1)$ and updates $B$ to $B'$, which includes the set of states that may be reached if action $a$ is taken at some state in $B$. 
 If the sensing action chosen by the P1 includes the hidden sensor, \ie,  $\sa_1 \in \SA_1 \setminus \SA_1^2$, or if the hidden sensor was queried earlier and P2 has not gained the capability to attack the hidden sensor yet, the value of $\detected$ is increased by $1$. Then, the nature player makes a probabilistic transition (represented by the dashed arrow) to a new state $s'$ according to the stochastic transition dynamics. P2 observes the new state and then chooses a sensor attack action $\sa_2$. If the hidden sensor has not been queried before or if P2 is not capable of reacting to the hidden sensor (\ie, $\detected < \tau_r$), P2 chooses a sensor attack action from $\SA_2^2$. If  P2 has the capability to attack the hidden sensor, P2 chooses an attack action from $\SA_2$. With the sensor attack by P2, P1 observes $\obs(s',\sa_1, \sa_2)$ and updates P1's belief to eliminate states in P1's belief that are not consistent with the observation.

 The belief-augmented game allows us to compute observation-based strategy as a belief-based one.
\\
 \begin{definition}[Belief-based Almost-Sure Winning Strategy/Region]
Given a two-player belief augmented game $\mathcal{H}$, a strategy $\pi_1$ is belief-based provided that for two states $(s,B,\detected),(s',B',\detected) \in Q_1$, if $B=B'$, then, $\pi_1((s,B,\detected)) = \pi_1((s',B',\detected))$. A set of states from which P1 has a belief-based, almost-sure winning strategy is called P1's almost-sure winning region with partial observations.\\
\end{definition}

 We now show that by solving the sensor-revealing game $\mathcal{H}$ to reach the final state, we can obtain the P1 strategy to satisfy the objective against P2 in the original game $G$. 
\\
\begin{lemma}
\label{lemma:belief-based-win-H-sufficient}
A belief-based almost-sure winning strategy to reach $q_F$ in P1's sensor-revealing game $\mathcal{H}$ is also almost-surely winning for P1 to visit $F$ in the zero-sum game with partially controllable observations $G$.
\end{lemma}
\begin{proof}
 
By construction of the game $\mathcal{H}$, the event of reaching the final state $q_F$ is conditioned on the event that the nature states $(s,B',a,\sa_1, \detected) \in Q_N$, where, $\post_{G}(s, a) \cap F \neq \emptyset$ or $\post_{G}(s, a) \subseteq F$ is visited. Let $Y \subseteq Q_N$ be all nature states that can be reached before visiting $q_F$. If $q_F$ is visited with probability one from any state in the almost-sure winning region, then the set $Y$ must be visited with probability one from any state in P1's almost-sure winning region in the game $\mathcal{H}$. Let $p=\underset{(s,B',a,\sa_1,\detected)\in Y}{\min}\mathsf{Pr}(F\mid s,a)$ be the minimal probability of reaching $F$ from taking action $a$ state $(s,B',a,\sa_1,\detected)$ in $Y$.  The probability of not reaching $F$ in $k$ visits to $Y$ is smaller than $(1-p)^k$. Also, if $F$ is not reached, then, the almost-sure winning strategy will reach some state $q'$ in the almost-sure winning region of P1 from which $Y$ is revisited with probability one. Hence, the probability of eventually reaching $F$ is $\underset{k\rightarrow \infty}{\lim}1-(1-p)^k = 1$. That is, an almost-sure winning strategy in $\mathcal{H}$ to reach $\{q_F\}$ is also almost-surely winning to visit $F$ in the game $G$.
\end{proof}
 
 
It is remarked that P2 may detect that the perceptual game $G^2$ is different from the true game $G$ even before observing P1 deploys the hidden sensor. In particular, if P2 detects that P1 has deviated from the equilibrium in perceptual game $G^2$, then she knows that her game model is incorrect. Lemma~\ref{lemma:belief-based-win-H-sufficient} shows that P1 winning in the sensor-revealing game $\mathcal{H}$ is a sufficient condition, but not a necessary condition. This is because the solution of $\mathcal{H}$ assumes that P2 knows the game played after revealing the hidden sensor, which is not the case.
 
To relax the assumption and assess the strategic value of a hidden sensor, we decompose the sensor-revealing game (Def.~\ref{def:sen_rev_game}) into two games: The first one, called \emph{initial game} is used to construct P1's game before P2 realizes that P2's perceptual game is not the true game. In addition, we restrict both players' actions to be \emph{subjective rationalizable actions} (to be defined next) in P2's perceptual game. The second game called \emph{delay-attack game}, is used to compute a set of states from which P1 can ensure to achieve his objective while P2 has complete knowledge about P1's sensing capability but has some delay in carrying out sensor attacks. Based on backward induction, the solution for the second game is used to define the objective for the first game. We thus start by solving the second game.
\\

 \begin{definition}[Delay-Attack Game]
 \label{def:delayed_attack_game}
 Given the sensor-revealing game $	\mathcal{H}$, the stochastic delayed-attack game is a tuple
      	\[
	\mathcal{G}^{+} = \langle Q^{+} \cup Q_{F}^{+}, (A_1 \times \SA_1) \cup \SA_2, \delta^+,\tau_r  \rangle,
	\]
where,
\begin{itemize}
    \item $Q^{+} = Q^{+}_1 \cup Q^{+}_N \cup Q^{+}_2$, is the set of states consisting of P1, nature states and P2's states. $Q^{+}_1 = \{(s,B,\detected) \in Q_1\mid  \detected \ge 0\}$ is the set of states where P1 selects a (control and sensing) action $(a, \sa_1)$. $Q^{+}_N = \{(s,B,a,\sa_1,\detected) \in Q_N \mid   \detected \ge 0  \}$ is the set of nature's state. $Q^{+}_2 = \{(s,B,\sa_1,\detected) \in Q_2 \mid  \detected \ge 0 \}$ is the set of states where P2 selects a sensor attack action. 
     \item $Q^{+}_F = \{q_F\}$ is the final state. It is also a sink state.

 \item The transition function $\delta^+$ can be obtained from the transition function $\delta$ in the game $\mathcal{H}$ by eliminating states that are not in $Q^+$ and their incoming/outgoing transitions.
    
    
\end{itemize}
The rest of the components are as in the game $\mathcal{H}$. 
 \end{definition}

The initial state is omitted because we will compute all winning initial states. Since the delay-attack game is a game with complete knowledge and one-sided (P1's) partial observations, we can employ the solution in \cite{lcss_paper}
 to solve the game.  Let the almost-sure winning region of P1 be represented as $\asw_{1}(\mathcal{G}^{+})$ and the almost-sure winning strategy for P1 be $\pi_1^+$. 

 \subsection{Computing the sensor-revealing strategy}
 In this section, we synthesize sensor-revealing almost-sure winning strategies for P1. 
 A key observation is that the revealing time step can be different from the time when P1 queries the hidden sensor.
The revealing time is when P2 detected a deviation of P1's strategy from the equilibrium in P2's perceptual game. To quantify this deviation, we first solve P2's perceptual game.
\\
\begin{definition}[Belief-augmented P2's perceptual game]
\label{def:P2's_perceptual_game}
    Given the P2's perceptual game $G^2$ and the sensor-revealing game $\mathcal{H}$, the belief-augmented P2's perceptual game is a tuple 
      	\[
	\mathcal{G}^{2} = \langle Q^{2} \cup Q_{F}, (A_1 \times \SA_1^2) \cup \SA_2^2, \delta^2, q_0 \rangle,
	\]

    where,
    \begin{itemize}
    \item $Q^{2} = Q^{2}_1 \cup Q^{2}_N \cup Q^{2}_2 $, is the set of states consisting of P1, nature and P2's states. $Q^{2}_1 = \{(s,B,\detected) \in Q_1 \mid  \detected =-1 \}$ is the set of states where P1 selects a (control and sensing) action $(a, \sa_1) \in \act_1^2$. $Q^{2}_N = \{(s,B,a,\sa_1,\detected) \in Q_N \mid   \detected =-1, (a,\sa_1)\in \act_1^2\}$ is the set of nature's state. $Q^{2}_2 = \{(s,B,\sa_1,\detected) \in Q_2 \mid   \detected =-1\}$ is the set of states where P2 selects a sensor attack action. 

     \item $Q_F = \{q_F\}$ is the final state. It is also a sink state.

\item $\delta^2$ is defined using $\delta$ by restricting the transition functions to the domain $(Q^2_1 \times  \act_1^2) \cup Q_N \cup (Q^2_2 \times \act_2^2) $.

    
\end{itemize}
The rest of the components are as in the game $\mathcal{H}$. 
\end{definition}

Note that the transitions for the game $\mathcal{G}^{2}$ are obtained from $\mathcal{H}$ by eliminating all states in which $\detected \ge 0$ (and their transitions) and all transitions enabled by P1's actions $\act_1 \setminus \act_1^2$.  We can also solve the above belief-augmented P2's perceptual game using the Algorithm \ref{alg:posg-reachability} from \cite{lcss_paper} to obtain the almost-sure winning region for P1 ($\asw_1(\mathcal{G}^{2})$) and P1's almost-sure winning strategy ($\pi_1^2$), perceived by P2.







 
 Using the solution of P2's perceptual game $\mathcal{G}^2$, we compute the \emph{subjectively rationalizable actions for two players} \cite{subjective_rationalizability_paper}. 
 \\
 \begin{definition}[Subjectively rationalizable actions for players in $\mathcal{G}^2$]
 \label{def:Subjective_rationalizability}
 Given the belief-augmented P2's perceptual game $\mathcal{G}^2$, 
let $\asw_i(\mathcal{G}^2)$ be the almost-sure winning region for player $i$ and  $\poswin_i(\mathcal{G}^2)$ is the positive-winning region for player $i$  in the game $\mathcal{G}^2$.
A \emph{subjectively rationalizable action} function $\mathsf{SR}_i: Q_i^{2} \rightarrow 2^{\act_i^2}$ maps each state $q \in Q_i^{2}$ to a set of actions of player $i$, for $i=1,2$ that are deemed rationalizable by P2 and is defined as follows. 
\[
\mathsf{SR}_i(q,\mathcal{G}^{2})= 
\begin{cases}
  \begin{aligned}
    &\{ a_i \in \act_i^2 \mid \post_{\mathcal{G}^2}(q,a_i) \subseteq \asw_i(\mathcal{G}^2)\}, \\ 
    &\hspace{9em} \text{if } q \in \asw_i(\mathcal{G}^2).
  \end{aligned}\\
  \begin{aligned}
    &\{ a_i \in \act_i^2 \mid \post_{\mathcal{G}^2}(q,a_i) \not\subseteq \asw_{j}(\mathcal{G}^2)\}, \\
    &\hspace{9em} \text{if } q \in \poswin_i(\mathcal{G}^2).
  \end{aligned}
\end{cases}
\]
    
      
 where $(i,j)\in \{(1,2), (2,1)\}$.

  \end{definition}


 In other words, from a state in the almost-sure winning region of P1, an action of P1  is considered rational by P2 if that action ensures the game stays within P1's almost-sure winning region; from a state in the positive winning region for P1, P1 can select an action which ensures the next state stays within P1's positive winning region in the game perceived by P2. P2's subjective rationalizable actions are defined similarly.

If P1 selects an action that is not subjectively rationalizable to P2, P2 would realize that P2's perceptual game is different from the true game even if that action does not query a hidden sensor. Therefore, a question arises: ``Is there any advantage for P1 to select a non-subjectively rationalizable action  without revealing the hidden sensor?'' 

Next, we demonstrate that if P1's objective is to achieve its goal with probability one, then the answer is \emph{no}.
To begin, we introduce the notions of revealing and sensor-revealing strategies.
\\
\begin{definition}[A Revealing Strategy]
\label{def:revealing_strategy}
Given the belief-augmented P2's perceptual game $\mathcal{G}^2$ and the sensor revealing game $\mathcal{H}$, a strategy $\pi_1: Q^2\rightarrow \dist(\act_1)$ is \emph{revealing} if for any strategy $\pi_2\in \Pi_2$, for at least one play $\rho  \in \outcomes_{\mathcal{G}^2}(q,\pi_1, \pi_2) $ starting from the state $q$, there exists a \emph{finite} time $0\le t<\infty$ such that $(a^t_1,\gamma_1^t) \in \act_1 \setminus \mathsf{SR}_1(q_t, \mathcal{G}^2)$,  where $q_t$ is the $t$-th state in this play and for all $0\le k <t$, $\pi_1(q_t,a)>0$ only if $a\in \mathsf{SR}_1(q, \mathcal{G}^{2})$. The time $t$ is called the \emph{revealing time}. 
Further, if $(a_1^{t},\sa_1^{t}) \in  \act_1 \setminus A_1^2$, then this revealing strategy is called the \emph{sensor revealing strategy}.
\\ 
\end{definition}

\begin{lemma}[pp. 23 in \cite{bertrand2017qualitative}]
\label{claim:P2_perfect_same_as_partial}
In the game $\mathcal{G}^2$, let $\asw_2(\mathcal{G}^2, \text{perfect})$ be P2's almost-sure winning region against a player P1 with \emph{perfect observations}. It holds that  \[\asw_2(\mathcal{G}^2) = \asw_2(\mathcal{G}^2, \text{perfect}).\]
\end{lemma}
The proof is provided in \cite{bertrand2017qualitative}.
  In a two-player stochastic game where P1 has a reachability objective and P2 has a safety objective, P2's almost-surely winning a safety game coincides with winning surely the safety game, which in turn coincides with winning surely against an opponent with perfect observations. This is because, for P2 to almost-surely satisfy it's safety objective from a state, then for every state in P2's almost-sure winning region, no action of P1 can force the game to leave P2's almost-sure winning region with a positive probability, regardless if P1 observes the state perfectly or not.
  \\
\begin{lemma}
\label{claim:P2_remains_in_P2s_winning}
    For any state $q=(s,B,-1) \in \asw_2(\mathcal{G}^2)$, there does not exist a strategy for P1 to satisfy P1's reachability objective with positive probability, or with probability one, regardless that strategy uses a hidden sensor or not.
\end{lemma}
\begin{proof}
Let's consider a strong P1$^*$ who has perfect observations. Then, from Lemma~\ref{claim:P2_perfect_same_as_partial} we have that for any state $q$ that is almost-sure winning for P2 against a weaker opponent P1 with partial observations is still almost-sure winning for P2 against P1$^*$. 

If P1 uses the hidden sensor, the best case scenario is for P1 to obtain perfect observation and become the strong P1$^*$. Thus, despite using the hidden sensor and becoming a strong P1$^*$, P1 cannot change $q$ to be P1's almost-sure winning or positive winning state.

\end{proof}

 \begin{lemma}
\label{lemma:not_using_hidden_sensor_is_losing}
    For any state $q=(s,B,-1) \in \poswin_1(\mathcal{G}^2)$, if P1 selects an action $(a,\sa_1) \notin \mathsf{SR}_1(q, \mathcal{G}^{2})$ and $\sa_1 \notin \SA_1\setminus \SA_1^2$, then P2 has a strategy to prevent P1 from satisfying its reachability objective with probability one.
\end{lemma}

\begin{proof}
Given that $q=(s,B,-1) \in \poswin_1(\mathcal{G}^2)$, for any $ (a, \sa_1) \in \act_1^{2}$, $\post_{\mathcal{G}^2}(q,a, \sa_1) \subseteq \poswin_1(\mathcal{G}^2) \cup \asw_2(\mathcal{G}^2)$ (by the solution of game $\mathcal{G}^2$, see \cite{lcss_paper}). If P1 selects $(a,\sa_1)\in \mathsf{SR}_1(q, \mathcal{G}^{2})$, then $\post_{\mathcal{G}^2}(q,a, \sa_1) \not\subseteq \asw_2(\mathcal{G}^2)$ by the definition of subjectively rationalizable actions.

Thus, while not taking the subjectively rationalizable actions and not using the hidden sensor, P1's action $(a,\sa_1) \in \act_1^{2} \setminus \mathsf{SR}_1(q, \mathcal{G}^2)$   ensures $\post_{\mathcal{G}^2}(q, a, \sa_1) \subseteq \asw_2(\mathcal{G}^2) $. From Lemma~\ref{claim:P2_remains_in_P2s_winning}, we have that from a state in $\asw_2(\mathcal{G}^{2})$, P1 has no strategy to satisfy his reachability objective with positive probability or with probability one regardless of using the hidden sensor or not.
Hence, if P1 takes an action $(a,\sa_1) \in \act_1^{2} \setminus \mathsf{SR}_1(q, \mathcal{G}^2)$, then P2  can ensure to win with probability one.
 \end{proof}

Lemma~\ref{claim:P2_remains_in_P2s_winning} and Lemma~\ref{lemma:not_using_hidden_sensor_is_losing} provide the following three insights: 
\begin{itemize}
	\item If $q\in \asw_1(\mathcal{G}^2)$, there is no need to use the hidden sensor because P1 can win without the hidden sensor.
	\item If $q\in \asw_2(\mathcal{G}^2) $, there is no need to use the hidden sensor because P1 cannot win even with the hidden sensor.
	\item If $q\in \poswin_1(\mathcal{G}^2)$, either P1 would take actions in $\mathsf{SR}_1(q, \mathcal{G}^2)$ to remain in the positive winning region or would take action $(a,\gamma_1)\in \act_1\setminus \act_1^2$ to reveal the hidden sensor, with a possibility to almost-surely win the game. That is,  P1 has no reason to deviate from the $\mathsf{SR}_1(q, \mathcal{G}^2)$ except when P1 decides to use the hidden sensor. 
	\end{itemize}

Based on these insights, we construct and solve P1's initial game to decide when to use the hidden sensor.
\\
 \begin{definition}[The Initial Game]
 \label{def:P2s_init_game}
Given P2's perceptual game $G^2$,  and players' subjectively rationalizable action functions $\mathsf{SR}_i$, $i=1,2$ (Def. \ref{def:Subjective_rationalizability}), and   sensor-revealing game $\mathcal{H}$ and P1's almost-sure winning region in the delay-attack game $\asw_1(\mathcal{G}^+)$,    P1's initial game is a stochastic two-player game with a reachability objective defined as:
      	\[
	\mathcal{G}^{0} = \langle Q^0 \cup \{q_F\}, \act_1, \act_2^2, \delta^0  \rangle,
	\]
where,
\begin{itemize}
    \item $Q^0 = Q^0_1 \cup Q^0_N \cup Q^0_2 $ where  $Q^0_1 = \{(s, B, \detected) \in Q_1 \mid   \detected =-1 \}$ is a set of P1's states, the set 
    $Q^0_N = \{(s, B, a, \sa_1, \detected) \in Q_N \mid  \detected \in \{-1,0\}\}$ is the set of nature's states. $Q^0_2 = \{(s, B, \sa_1, \detected) \in Q_2 \mid \detected \in \{-1\}\}$ is the set of P2 states.
    
    \item $q_F$ is the unique final state in $\mathcal{G}^0$ and is a sink state.
    \item $ \act_1 $ is the set of actions for P1.
    \item $\act_2^2$ is the set of actions for P2.   

\item $\delta^0: Q^0\times \act_1\times \act_2^2\rightarrow \dist(Q^0)$ is the probabilistic transition function.
 For any state $q\in \asw_1(\mathcal{G}^2)$, let $\delta^0(q,q_F)=1$, which means $q_F$ is reached with probability one. 
 And for any state $q \in \asw_2(\mathcal{G}^2)$, for any action $a  \in \act_1 \cup \act_2^2$, let $\delta^0(q,a,q) = 1$. For any state $q\notin \asw_1(\mathcal{G}^2) \cup  \asw_2(\mathcal{G}^2)$, the transition is defined as follows. 
\begin{enumerate}
\item Given a P1's state $q = (s,B,-1) \in Q_1^{0}$, we have the following two cases,
        \begin{enumerate}
            \item For an action $(a,\sa_1) \in \mathsf{SR}_1(q,\mathcal{G}^2)$, $\delta^0(q,(a,\sa_1), q')= \delta(q,(a,\sa_1),q')$ for any $q'$ reachable from $q$. 
            \item For an action $(a,\sa_1) \in \act_1\setminus \act_1^2$, $\delta^0(q,(a,\sa_1),(s,B',a,\sa_1,0))=1$  where $B'=\post_{G}(B,a)$.
         \end{enumerate} 

\item Given a nature's state $q=(s,B,a,\sa_1,\detected) \in Q_N^{0}  $, we have the following three cases,
\begin{enumerate}[I.]
    \item $\detected=-1$, let $\delta^{0}(q,q') = \delta(q,q')$ for any $q'$ reachable from $q$.
    \item  $\detected =0$  and $q \in \asw_1(\mathcal{G}^{+})$, let $\delta^{0}(q,q_F) = 1$. That is, from a state in P1's almost-sure winning region in the delay attack game, a transition to $q_F$ is made with probability one.
    \item $\detected=0$ and $q \not \in \asw_1(\mathcal{G}^{+})$, $\delta^{0}(q,q)=1$.  If the sensor is revealed to P2 and the state is not winning for P1 in the delay-attack game, then the state becomes a sink state. 
\end{enumerate}

\item Given a P2's state $q = (s,B,\sa_1,-1) \in Q_2^{0}$, 
only actions in $\mathsf{SR}_2(q,\mathcal{G}^2)$ is allowed and for any $\sa_2 \in \mathsf{SR}_2(q,\mathcal{G}^2)$, 
$\delta^{0}(q,\sa_2,q') = \delta(q, \sa_2, q')$ for any reachable $q'$.
 
    \end{enumerate}
    \item P1's objective in the initial game is to reach the state $q_F$.
\end{itemize}
\end{definition}

P1's initial game can now be solved for P1's almost-sure winning using the Algorithm~\ref{alg:posg-reachability} from \cite{lcss_paper} and thus, we obtain $\asw_1(\mathcal{G}^0)$ and $\pi_1^0$. 

Because P2's perceptual game changes during the interaction, we introduce the notion of behavior subjective rationalizable strategy profile as the solution concept of the dynamic hypergame. 
\\
\begin{definition}[Behavior Subjectively Rationalizable Strategy]
A strategy $\pi_2\colon Q \to \dist (\act_2)$ is behaviorally subjectively rationalizable for P2 if, 
  \[
   \supp (\pi_2(q )) \subseteq \begin{cases}
        \mathsf{SR}_2(q, \mathcal{G}^2), & \text{if } \hspace{2mm} \last (q) = -1,\\
        \mathsf{SR}_2(q, \mathcal{G}^{+}), & \text{if } \hspace{2mm} {\last(q) \ge 0 .}
    \end{cases} \]
  where $\last(q)$ maps a state $q$ to the last component $\detected$ of $q$, $\mathsf{SR}_2(q, \mathcal{G}^2) $ is a subjectively rationalizable strategy of P2 in P2's perceptual game $\mathcal{G}^2$ and $\mathsf{SR}_2(q,\mathcal{G}^{+})$ is a subjectively rationalizable strategy of P2 in the delay attack game $\mathcal{G}^+$ given analogously as,
\[
\mathsf{SR}_2(q, \mathcal{G}^{+})= 
 \begin{cases}
  \begin{aligned}
    &\{\sa \in \SA_2 \mid \post_{\mathcal{G}^{+}}(q,\sa) \subseteq \asw_2(\mathcal{G}^{+})\}, \\
    &\hspace{8.5em} \text{if } q \in \asw_2(\mathcal{G}^{+}).
  \end{aligned}\\
  \begin{aligned}
    &\{\sa \in \SA_2 \mid \post_{\mathcal{G}^{+}}(q,\sa) \not\subseteq \asw_1(\mathcal{G}^{+})\}, \\
    &\hspace{8.25em} \text{if } q \in \poswin_2(\mathcal{G}^{+}).
  \end{aligned}
 \end{cases}
\]


  
\end{definition}
In words, P2 always commits to a rational strategy in her perceptual game, which is initially $\mathcal{G}^2$ and then changes to $\mathcal{G}^+$ after detecting the use of a hidden sensor. 

 In the following lemma, we show that the almost-sure winning strategy computed by P1 remains winning even if P2 does not choose all subjectively rationalizable actions with a positive probability. 
\\
 \begin{lemma}
    Any almost-sure winning strategy $\pi_1 \in \Pi_1$ for P1 in the games $\mathcal{G}^{0}$ or $\mathcal{G}^{+}$ when P2 plays a strategy such that $\supp(\pi_2(q)) = \mathsf{SR}_2(q, \mathcal{G}^{2})$ or $\mathsf{SR}_2(q, \mathcal{G}^{+})$, respectively, for all $q$ where $\pi_2$ is defined, remains almost-sure winning for P1 when P2 employs a strategy such that there exists $q$ where $\pi_2(q)$ is defined, $\supp(\pi_2(q)) \subset \mathsf{SR}_2(q, \mathcal{G}^{2})$ or  $\supp(\pi_2(q)) \subset \mathsf{SR}_2(q, \mathcal{G}^{+})$.
\end{lemma}

\begin{proof}
Let us consider a P2$^\dagger$ who follows the strategy $\pi_2$ that satisfies, for at least one $q$ where $\pi_2$ is defined, $$\supp(\pi_2(q)) \subset    \begin{cases}
        \mathsf{SR}_2(q, \mathcal{G}^2), & \text{if } \hspace{2mm} \last (q) = -1,\\
        \mathsf{SR}_2(q, \mathcal{G}^{+}), & \text{if } \hspace{2mm} {\last(q) \ge 0, } \end{cases}$$ 
        for a state $q\in Q^0$ or $q\in Q^+$. That is, P2 follows a strategy that has a distribution over a strict subset of subjectively rationalizable actions. From Def.~\ref{def:delayed_attack_game} and Def.~\ref{def:P2s_init_game}, we construct the game $\mathcal{G}^{+}$ and $\mathcal{G}^{0}$ considering all the subjectively rationalizable actions for P2. Thus, the delayed-attack game and initial game constructed with P2$^\dagger$ would be a subgame of $\mathcal{G}^{+}$ and $\mathcal{G}^{0}$ respectively, thereby making P2$^\dagger$ a weak opponent to P1. 
Since P1's strategy $\pi_1$ is winning in $\mathcal{G}^0$ (or $\mathcal{G}^+$) against a stronger P2, it remains winning against the weaker opponent P2$^\dagger$.

\end{proof}

 Next, we introduce the notion of a deceptive almost-sure winning strategy and provide Lemma \ref{lemma:dasw_1_is_supset_asw} and Lemma \ref{lemma:almost_sure_in_G2_is_subset_G0} to demonstrate that the deceptive almost-sure winning region in the sensor-revealing game can be determined by computing the almost-sure winning regions of both the delay-attack game and initial game.

\begin{definition}[Deceptive Almost-Sure Winning Strategy/Region] 
\label{def:dasw_strategy}
In the sensor-revealing game $\mathcal{H}$,  a state $q\in Q$ is \emph{\ac{dasw}} for P1 if there exists a strategy for P1 that ensures with probability one, $q_F$ can be reached against any behavior subjectively rationalizable strategy of P2. The set of deceptive almost-sure winning states in $\mathcal{H}$ is called the \emph{\ac{dasw} region} for P1, denoted $\dasw_1(\mathcal{H})$.\\
  \end{definition}

\begin{lemma}
    \label{lemma:dasw_1_is_supset_asw}
$\dasw_1(\mathcal{H})\supseteq \asw_1(\mathcal{G}^{0}) \cup \asw_1(\mathcal{G}^{+})$.
\end{lemma}
\begin{proof}
Consider states $q \in Q^{+}$ in the game $\mathcal{G}^{+}$. It follows from the construction of $\mathcal{G}^{+}$ that $\last(q) \in \{0,1,\ldots, \detected \}$. Additionally, the transitions for each state $q$ in $\mathcal{G}^{+}$ follow the same transitions as in $\mathcal{H}$ since the actions available for P1 and P2 are the same in both games. Thus, we conclude that  $\asw_1(\mathcal{G}^{+}) \cap Q^{+} = \dasw_1(\mathcal{H}) \cap Q^{+}$.

Next, consider states $q \in Q^{0}$. In this case, we have $\last(q) \in \{-1,0\}$. For every state $q$ such that  $q \in Q^{0}$ with $\last(q) = -1$, the enabled actions and outgoing transitions in $\mathcal{H}$ and $\mathcal{G}^0$ are the same. For any strategy profile $ (\pi_1,\pi_2)$, 
let $  \mathsf{Pr}^{(\pi_1,\pi_2)}( \mathsf{Reach}( q_F) \mid q,a, \mathcal{G}^0)$ (resp. $\mathsf{Pr}^{(\pi_1,\pi_2)}( \mathsf{Reach}( q_F) \mid q,a, \mathcal{H})$) denote the probability of reaching the final state $q_F$ of $\mathcal{G}^0$ (resp. $\mathcal{H}$) given the state $q$ and action $a$ in the initial game $\mathcal{G}^0$ (resp. $\mathcal{H}$).

For states $q$ with $\last(q) = 0$, then if $q\in \asw_1(\mathcal{G}^{+})$, a transition to $q_F$ is made with probability one. Otherwise, $q$ is made a sink state.  Thus, for a state $q \notin \asw_1(\mathcal{G}^{+})$ and $\last(q)=0$, $\mathsf{Pr}^{(\pi_1,\pi_2)}( \mathsf{Reach}( q_F) \mid q,a, \mathcal{G}^0)=0$.  For the same state in $\mathcal{H}$, it is not a sink state and thus $\Pr^{(\pi_1,\pi_2)}(\mathsf{Reach}(q_F)\mid q,a, \mathcal{H}) \ge 0$. Based on backward induction, it holds that $\Pr^{(\pi_1,\pi_2)}(\mathsf{Reach}(q_F)\mid q,a, \mathcal{H}) \ge  \Pr^{(\pi_1,\pi_2)}(\mathsf{Reach}(q_F)\mid q,a, \mathcal{G}^0)$. Thus $\asw_1(\mathcal{G}^{0}) \subseteq \dasw_1(\mathcal{H})$. 


Taking into account the $\asw_1(\mathcal{G}^{0}) \subseteq \dasw_1(\mathcal{H})$,  $\asw_1(\mathcal{G}^{+}) \subseteq \dasw_1(\mathcal{H})$, it follows that $\dasw_1(\mathcal{H}) \supseteq \asw_1(\mathcal{G}^{0}) \cup \asw_1(\mathcal{G}^{+})$.\\
\end{proof}

\begin{lemma}
    \label{lemma:almost_sure_in_G2_is_subset_G0}
    $\asw_1(\mathcal{G}^{2}) \subseteq \asw_1(\mathcal{G}^{0})$.
\end{lemma}
\begin{proof}
 By the construction of $\mathcal{G}^0$, if $q\in \asw_1(\mathcal{G}^2)$, a transition to $q_F$ is made with probability one. As a result, $q\in \asw_1(\mathcal{G}^0)$. Thus, we have that $\asw_1(\mathcal{G}^2) \subseteq \asw_1(\mathcal{G}^0)$.\end{proof}
$\\$
\begin{theorem}
\label{claim:winning_region_H}
The deceptive almost-sure winning region in the game $\mathcal{H}$ satisfies
    $$\dasw_1(\mathcal{H}) = \asw_1(\mathcal{G}^0) \cup \asw_1 (\mathcal{G}^{+}).$$
\end{theorem}

\begin{proof} 
        We need to show $\dasw_1(\mathcal{H}) \subseteq \asw_1(\mathcal{G}^0) \cup \asw_1 (\mathcal{G}^{+})$. By way of contradiction, suppose there exists a state $q \in \dasw_1(\mathcal{H})  \setminus (\asw_1(\mathcal{G}^0) \cup \asw_1 (\mathcal{G}^{+}))$, then it must be $\last(q)=-1$ because $\dasw_1(\mathcal{H}) \cap Q^+ =  \asw_1 (\mathcal{G}^{+}) \cap Q^+$ (see the proof of Lemma~\ref{lemma:dasw_1_is_supset_asw}).
        
        Let $\pi_1^\dagger$ be P1's DASW strategy, for any P2's behavioral subjectively rationalizable strategy $\pi_2$, 
        the set of runs in $\outcomes_{\mathcal{H}}(q, \pi_1^\dagger, \pi_2)$ satisfy: 1) for any $\rho =q_0q_1\ldots \in \outcomes_{\mathcal{H}}(q, \pi_1^\dagger, \pi_2)$, there exists $k$ such that $q_k\in F$; and 2)
          at least one run $\rho^\ast = q_0q_1\ldots  \in \outcomes_{\mathcal{H}}(q,\pi_1^\dagger, \pi_2)$ such that there exists an $i\ge 0$,  $\last(q_i) = 0$ and $q_i\in \dasw_1(\mathcal{H})\cap Q^+$ and for all $0\le j < i$, $\last(q_j)=-1$. Because if the run $\rho^\ast$ does not exist, then $q\in \asw_1(\mathcal{G}^2) \subseteq \asw_1(\mathcal{G}^0)$.   From Lemma \ref{lemma:almost_sure_in_G2_is_subset_G0}, it contradicts the assumption that $q\notin \asw_1(\mathcal{G}^0)$).
          
In the first case, if $q_k\in F$, the same run is feasible in $\mathcal{G}^0$ and state $q_F$ is reached with probability one in $\mathcal{G}^0$ next. In the second case, because  $q_i\in \dasw_1(\mathcal{H})\cap Q^+$, and $\dasw_1(\mathcal{H})\cap Q^+ = \asw_1(\mathcal{G}^+)\cap Q^+ $, a transition to $q_F$ is also made with probability one in $\mathcal{G}^0$ next. Thus, $\pi_1^\dagger$ with the domain restricted to $Q^0$ is an almost-sure winning strategy for P1 in $\mathcal{G}^0$,  contradicting the assumption that $q\notin \asw_1(\mathcal{G}^0)$.


Thus, with $\dasw_1(\mathcal{H}) \subseteq \asw_1(\mathcal{G}^0) \cup \asw_1(\mathcal{G}^+)$ and Lemma~\ref{lemma:dasw_1_is_supset_asw}, the proof is completed.
\end{proof}

 \textbf{Strategy construction:} Given the winning regions  $\asw_1(\mathcal{G}^0)$, $\asw_1(\mathcal{G}^2)$ and $\asw_1(\mathcal{G}^+)$ and their respective winning strategy for P1 $\pi_1^0, \pi_1^2, \pi_1^+$, 
we can construct the  sensor revealing strategy for P1  $\pi_1^{\ast}:(\asw_1(\mathcal{G}^0) \cup \asw_1(\mathcal{G}^{+})) \rightarrow \dist(\act_1)$ such that
for each $q \in \dasw_1(\mathcal{H})$,

  \[
    \pi_1^{\ast}(q) =\begin{cases}
        \pi_1^{2}(q), & \text{if }  q\in \asw_1(\mathcal{G}^{2}), \\
        \pi_1^{0}(q), & \text{if }  q\in \asw_1(\mathcal{G}^{0})\setminus   \asw_1(\mathcal{G}^{2}), \\
        \pi_1^{+}(q), & \text{if }  q\in \asw_1(\mathcal{G}^{+}).
    \end{cases} \]

We use the running example to illustrate using the sensor-deception strategy to transform a positive winning state for P1 to almost-sure winning.
\\
\begin{example}
\label{example:advantage_with_hidden_sensor_tau_1}
Consider the example in Fig.~\ref{fig:running_example_1} again and let P2's delay of reaction $\tau_r = 1$. First, we construct the delay attack game $\mathcal{G}^{+}$. On solving the delay attack game, we obtain the $\asw_1(\mathcal{G}^{+})$. With the solution of $\mathcal{G}^{+}$, we construct the initial game $\mathcal{G}^0$.  A fragment of the construction of delay attack game $\mathcal{G}^{+}$ and the initial game $\mathcal{G}^0$  is as shown in Fig.~\ref{fig:positive_trans_almost_example}. In Fig.~\ref{fig:positive_trans_almost_example}, the perception actions of P1 are $\sa_1^0 = \{A,B\}$, $\sa_1^{1} = \{D,C\}$ and $\sa_1^2=\{A,D\}$.

\begin{figure}[h!]
  \captionsetup[subfigure]{aboveskip=-20pt,belowskip=0pt}
    \centering       
    \begin{subfigure}[b]{0.7\textwidth}
\begin{tikzpicture}[->,>=stealth',shorten >=1pt,auto,node distance=2.3cm,scale=.58,semithick, transform shape,square/.style={rectangle}]
		\tikzstyle{every state}=[fill=black!10!white];

  \node[state,diamond, aspect = 2.5] (0) 
    at (0, 0) 
    {$(s_1,\{s_2,s_3\},a,\sa_1^0,0)$};
  
  \node[state, rectangle] (1)
    at (5.75, 0) 
    {$(s_3,\{s_2,s_3\},0)$};

  \node[state, ellipse] (2) 
    at (10, 0) 
    {$(s_3,\{s_3\},0)$};
  
  \node[state, diamond, aspect=3] (3) 
    at (10, -2)
    {$(s_5,\{s_5\},a,\sa_1^1)$};
    
  \node[state, rectangle] (4) 
    at (5, -2)
    {$(s_5,\{s_5\},1)$};
    
  \node[state, ellipse] (5) 
    at (0, -2)
    {$(s_5,\{s_5\},1)$};
    
  
    
    

  
  \path (0) edge[dashed] node{} (1)
        (1) edge  node{$\{A\}$} (2)
        (2) edge  node{$(a, \sa_1^1)$}(3)
        (3) edge [dashed]  node{}(4)
        (4) edge node{$\{C\}$} (5)
        (5) edge[loop below] node{$1$} (5)        
  ;

\end{tikzpicture}
        \caption{}
        \label{fig:a}
    \end{subfigure}

    \begin{subfigure}[b]{0.7\textwidth}
\begin{tikzpicture}[->,>=stealth',shorten >=1pt,auto,node distance=2.3cm,scale=.58,semithick, transform shape,square/.style={rectangle}]
		\tikzstyle{every state}=[fill=black!10!white];

  \node[state,ellipse] (0) 
    at (0, 0) 
    {$(s_1,\{s_1\},-1)$};
  
  \node[state, diamond, aspect = 3] (1)
    at (5.75, 0) 
    {$(s_1,\{s_2,s_3\},a, \sa_1^{0}, 0)$};

  \node[state, accepting] (2) 
    at (10, 0) 
    {$q_F$};
  
  \node[state, diamond, aspect=3] (3) 
    at (0, -2.5)
    {$(s_1,\{s_2,s_3\},a,\sa_1^{2},-1)$};
    
  \node[state, rectangle] (4) 
    at (5.35, -2.5)
    {$(s_2,\{s_2,s_3\},-1)$};
    
  \node[state, ellipse] (5) 
    at (10, -2.5)
    {$(s_2,\{s_2,s_3\},-1)$};
    
  \node[state,rectangle] (6) 
    at (0, -5)
    {$(s_3,\{s_2,s_3\},-1)$};
  
  \node[state,ellipse] (7) 
    at (5.35, -5)
    {$(s_3,\{s_2,s_3\},-1)$};
    
    

  
  \path (0) edge node{$(a,\sa_1^{0})$} (1)
        (1) edge[dashed]  node{} (2)
        (2) edge [loop above] (2)
        (0) edge node{$(a,\sa_1^{2})$} (3)
        (3) edge[dashed] node{} (4)
        (4) edge node{$\{D\}$} (5)
        (3) edge[dashed] node{} (6)
        (6) edge node{$\{D\}$} (7)

  ;

\end{tikzpicture}
        \caption{}
        \label{fig:b}
    \end{subfigure}
 \caption{
        \textbf{(a)} A fragment of the constructed $\mathcal{G}^{+}$.
        \textbf{(b)} A fragment of the constructed $\mathcal{G}^0$.
        In the above fragments of games, oval/circle states are P1's states, rectangle states are P2's, and diamond states are nature states.
    }
    \label{fig:positive_trans_almost_example}
\end{figure}
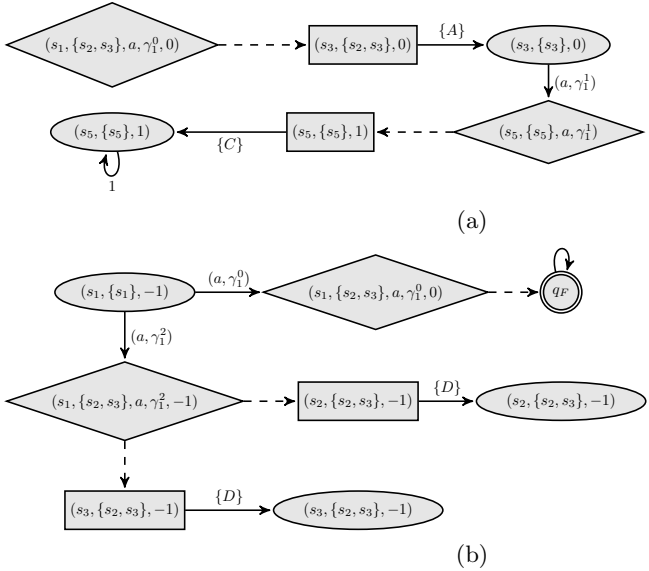
 
	Consider that P1 starts from the state $(s_1, \{s_1\},\detected = -1) \in \mathcal{G}^0$, which a positive winning state for P1 in the original game. Let us first consider the case when P1 would play without using the hidden sensor, then,  P1 could take an action $(a,\{A,D\})$. Here, the perception-action $\{A,D\}$ is the best sensor query for P1. And with the current action, the state transitions to $(s_1,\{s_2,s_3\},a,\{A,D\},-1)$. With nature's state choosing a next state stochastically, we end up in a P2 state $(s_3,\{s_2,s_3\},\{A,D\},-1)$ or $(s_2, \{s_2, s_3\}, \{A, D\}, -1)$. If P2 chooses to attack the sensor $D$, the system transitions to  P1 state $(s_3, \{s_2,s_3\},-1)$ or $(s_2, \{s_2, s_3\}, -1)$. Either P1 state is a losing state because from state $s_2$, action $b$ is winning and $a$ is losing and vice versa from state $s_3$. There is no consistent winning action for the given belief.\\
	If P1 were to use a hidden sensor, then, at the state $(s_1,\{s_1\},\detected=-1) \in \mathcal{G}^0$, P1 chooses an action $(a,\{A,B\})$ which includes the hidden sensor $B$. Then, the state transitions to the nature's state $(s_1,\{s_2,s_3\},a,\{A,B\},0)$ as shown in the Fig.~\ref{fig:positive_trans_almost_example}(b). Next, since the nature state reached is an almost-sure winning state for P1 in the game $\mathcal{G}^{+}$, as shown in the fragment in Fig.~\ref{fig:positive_trans_almost_example}(a), there is a direct transition to the winning state $q_F$. 
 Looking at the game $\mathcal{G}^{+}$, we can see that from the nature state $(s_1,\{s_2,s_3\},a,\{A,B\},0)$, nature selects stochastically one of the two next possible states $(s_3,\{s_2,s_3\},\{A,B\},0)$ or $(s_2,\{s_2,s_3\},\{A,B\},0)$ at which P2 selects an attack action.   Given that $\detected =0 < \tau_r$, P2 cannot attack the hidden sensor and can only attack the other sensor \ie~ $A$. This attack results in a P1's state $(s_3, \{s_3\},0)$ or $(s_2, \{s_2\}, 0)$. It can be seen that both states are almost-sure winning for P1.  \\
Thus, the state $(s_1, \{s_1\}, -1)$, a positive winning state in the original game, is transformed into an almost-sure winning state using sensor deception. 
\\
\end{example}

Next, we are interested in determining the deceptive almost-sure winning region and comparing it with the non-deceptive almost-sure winning region. To do this, we formally introduce the notion of the value of deception.  Specifically, we consider two notions of the value of deception: one for individual states and one for the game as a whole. Informally, the value of deception for a state indicates whether that state is winning \emph{only} on using the hidden sensor. The value of deception for a game compares the deceptive almost-sure winning region with the non-deceptive almost-sure winning region. 
\\
\begin{definition}[Value of Deception]
    Given the assumption \ref{assumption:hidden_sensor}, P1 is said to have an advantage at a state $q \in Q_1^{0}$ if the value of deception is $1$. The Value of Deception (VoD) is a boolean function, as given below.
    \[
     VoD(q) = \begin{cases}
     \begin{aligned}
         & 1, \\& \hspace{1em}\text{for }  q\in \asw_1(\mathcal{G}^{0})\setminus   \asw_1(\mathcal{G}^{2}) \land \\&\hspace{7em} (\supp(\pi_1^0(q))\cap \act_1^{2} = \emptyset), \end{aligned}\\
         \begin{aligned}
        & 0 , & \text{otherwise }. \end{aligned}\\
     \end{cases}
    \] 
 In other words, the value of deception is $1$ for any P1 state in the game $\mathcal{G}^{0}$ if the state transforms to an almost-sure winning state from a positive winning state only on using the hidden sensor. The overall value of deception for the game, \ie, the value of deception as extended to a game, is as follows.
\[
VoD = \dfrac{|\poswin_1(\mathcal{G}^2)|-|\poswin_1(\mathcal{G}^0)|}{|\poswin_1(\mathcal{G}^2)|}
\]
where, 
    \begin{itemize}
        \item $|\poswin_1(\mathcal{G}^2)|$ represents the number of positive winning states for P1 in the game $\mathcal{G}^2$.
        \item $|\poswin_1(\mathcal{G}^0)|$ represents the number of positive sure winning states for P1 in the game $\mathcal{G}^0$.
    \end{itemize}
    \end{definition}
    In words, the VoD of the game measures the fraction of positive winning states of P1 being changed into deceptive almost-sure winning.

\section{Experimental Case Studies}
\label{sec:Experimental_case_studies}

In this section, we illustrate the developed theory with two practical examples. Initially, we utilize a stochastic graph network scenario to highlight the advantage of using a hidden sensor. Subsequently, we employ it in a stochastic Mario Bros gridworld environment. All simulation experiments were implemented using Python on a Windows 11 platform with a $3.2$$GHz$ processor and $32$$GB$ memory
\footnote{The code for the illustrative example and the Mario Bros. example is available on \url{https://github.com/Sumukha-Udupa/Reactive_synthesis_of_sensor_revealing_strategies}.}.

\subsection{Illustrative example.}


In this illustrative example, consider a graph network as shown in Fig.~\ref{fig:Example1_graph}. The graph network comprises $25$ states and the objective for P1 is to reach one of the target states while avoiding the losing sink states, namely nodes $3, 9$, and $16$ (colored red). Additionally, states $11$ and $12$ are the final target states (colored green). The edges in the graph denote permissible actions from each state. The set of all available actions are $A_1 = \{\mathsf{D1, D2, D3, D4, D5, D6}\}$. An edge labeled with an action $a$ from one state  $s$ to another state $s'$ means that by taking that action $a$ from state $s$, it is possible to reach $s'$. For instance, on taking the action $\mathsf{D1}$ at the state $0$, the state transitions to $7$ or $1$.  The graph only shows edges/transitions with non-zero probabilities and the exact probabilities are omitted.  For qualitative planning, we only need to know if there is a positive probability from one state to another given the action and do not need to know the exact probability.

P1 has deployed a sensor network of $5$ sensors $\mathsf{A,B,C,D,E}$. Each sensor monitors a set of specific nodes. Sensor $\mathsf{A}$ observes the nodes $\{\mathsf{4}, \mathsf{10}, \mathsf{13}, \mathsf{15}, \mathsf{17}, \mathsf{19}\} $, while  $\mathsf{C}$ monitors $\{\mathsf{3, 9, 11, 12, 16, 18, 20, 21, 22, 23, 24}\}$. Sensor $\mathsf{D}$ covers $\{\mathsf{3, 4, 7, 9, 10, 17, 19, 21, 22, 23, 24 }\}$, and $\mathsf{E}$ is responsible for $\{\mathsf{1, 2, 5, 6, 8, 14}\}$. Finally, sensor $\mathsf{B}$ oversees the states  $\{\mathsf{4, 5, 6, 15, 18, 20,} \\ \mathsf{23, 24}\}$. All of the sensors are Boolean sensors which return a Boolean value $\mathsf{True}$ when P1 is within the sensor's range and return $\mathsf{False}$ otherwise. This sensor network is susceptible to attack by the adversary P2. 

\begin{figure}[ht!]
        \centering       
       \includegraphics[scale=0.34]{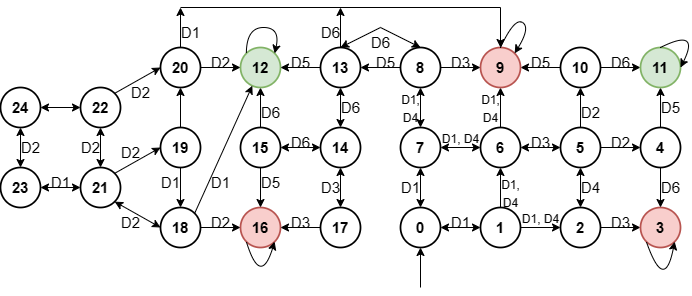}
        \caption{Graphical representation of the adversarial environment.
        }
      \label{fig:Example1_graph}
\end{figure}

\textbf{Example Setup.} To effectively illustrate the advantage of incorporating a hidden sensor into planning paths, we first establish a baseline by experimenting with all sensors known to both players. 

We then consider the scenario where only the sensor $\mathsf{B}$ operates as the hidden sensor. We denote the delay in reaction to the hidden sensor by P2 as $\tau$.

\begin{enumerate}[]
    \item \textbf{Case 1}: \emph{The setup with no hidden sensor.}
    
    In this configuration, P2 is aware of all the sensors. P1 can query any three sensors at a time to determine its position. And P2 can attack any one of the sensors at a time. 
    
    \item \textbf{Case 2}: \emph{The setup with a hidden sensor.}
    
    In this configuration, P1 can query any three sensors simultaneously. The sensor $\mathsf{B}$ in the sensor network is hidden. Initially, the adversary can only attack one of the four known sensors, namely, $\mathsf{A,C,D,E}$. However, once P1 queries the hidden sensor, P2 gains the ability to attack any of the available sensors after a delay  $\tau$.

\end{enumerate}

The experimental results for the cases mentioned above, where we vary the delay parameter ($\tau$), are summarized in Table~\ref{tab:1}. 
\\
\begin{table}[h!]      \vspace{-1ex}

     \caption{Results of experimentation with Case 1 (without a hidden sensor) and Case 2 (with hidden sensor)  by varying the delay ($\tau$) for the illustrative example.}
    \vspace{1ex}
    \label{tab:1}
 \setlength\tabcolsep{1pt}
        \begin{tabular}{|l|c|c|c|}
        \hline

\textbf{Setup ($\tau$)}&
\textbf{Deceptive Winning Initial States}& \textbf{VoD} \\
\hline
 Case 1 (-) 
&$\mathsf{[4, 8, 10, 11, 12, 13, 15, 18, 20]}$ &$0$ \\ \hline  Case 2 (0) 
&$\mathsf{[4, 8, 10, 11, 12, 13, 15, 18, 20]}$ &$0$ \\ \hline

 Case 2  (1) 
&$\mathsf{2, 5, 6, [4, 8, 10, 11, 12, 13, 15, 18, 20]}$ &$0.25$ \\ \hline
 Case 2 (2) 
&$\mathsf{2, 5, 6, [4, 8, 10, 11, 12, 13, 15, 18, 20]}$ &$0.25$ \\ \hline
 Case 2 (3) 
&$\mathtt{0, 1, 2, 5, 6, 7, [4, 8, 10, 11, 12, 13, 15, 18, 20]}$ &$0.50$ \\ \hline

    \end{tabular}
 \end{table}

\subsubsection{Discussion}
The results of the experiments on the illustrative example tabulated in Table~\ref{tab:1} show the set of initial states from which P1 has a deceptive $\asw$  strategy, given the best response sensor attack strategy by P2. These states are referred to as the winning initial states. To obtain the winning initial states with deception, we construct P2's perceptual game (Def. \ref{def:P2's_perceptual_game}), the delayed attack game (Def. \ref{def:delayed_attack_game}) and then P1's initial game (Def. \ref{def:P2s_init_game}). 

We also observe from Table~\ref{tab:1} that the number of winning initial nodes is the same for Cases 1 and 2 when there is no hidden sensor and when $\tau_d=0$. 

When P1 has a hidden sensor (sensor $\mathsf{B}$) and P2's delay in attack is greater than 0, P1's winning initial states are also shown in Table \ref{tab:1}. 
It's noteworthy that even with a small delay, such as $\tau = 1$, there is an evident increase in the number of winning initial nodes for P1.

To see why the delay helps, let us consider the game starting from node $\mathsf{5}$. From node $\mathsf{5}$, P1 can reach either $\mathsf{10}$ and $\mathsf{4}$ with action $\mathsf{D2}$ and is unsure which one is reached given the partial observation. However, with a belief $\{\mathsf{4, 10}\}$, P1 has no consistent actions to reach the goal state $\mathsf{11}$. If P1 employs the hidden sensor and P2 has a delay $\tau \ge 1$, P1 can distinguish which nodes it is at thereby making $\mathsf{5}$ a winning initial state.  Given that the only winning action in the state $\mathsf{5}$ is to query the hidden sensor, $VoD(\mathsf{5})=1$. 

Furthermore, we observe that an almost-sure winning node for P2 remains winning for P2. For example, if the game starts from node  $\mathsf{17}$  where P2 has an almost-sure winning strategy in the game $\mathcal{G}^2$ where P2 does not know the presence of hidden sensors. The same state remains to be almost-sure winning even if P1 can employ the hidden sensor. This observation aligns with  Lemma~\ref{claim:P2_remains_in_P2s_winning}. 

For Case $\tau=0$, the construction times for games $\mathcal{G}^2$, $\mathcal{G}^{+}$, and $\mathcal{G}^0$ are $0.07s$, $0.2s$, and $0.002s$, respectively. The overall time to construct and solve all the games to obtain the policy is $1.136s$. Similarly, for Case $\tau=1$, the construction times for these games remain the same, and the overall time is $1.312s$. In Case $\tau=2$, the construction times are $0.069s$, $0.28s$, and $0.002s$ for $\mathcal{G}^2$, $\mathcal{G}^{+}$, and $\mathcal{G}^0$ respectively, with $1.917s$ as the overall time to solve to obtain the policy. Finally, for Case $\tau = 3$, the solve time to obtain the policy is $3.1081s$, and the construction times for the three games are $0.069s$, $0.41s$, and $0.002s$.

\subsection{Mario Bros. Example.}

In this example, we delve into a scenario inspired by the iconic game Mario Bros., set within a stochastic $\mathsf{6}\times\mathsf{6}$ gridworld environment, as illustrated in Fig.~\ref{fig:mario_bros_example}. Just as in the game, our protagonist, Mario, endeavors to navigate this gridworld to reach the Princess, while steering clear of perilous Bombs and the menacing Piranha Plant, both losing sink states. The adversary within this environment mirrors the notorious King Koopa from the original game, with the objective of thwarting Mario from accomplishing his task.

\begin{figure}[ht!]
        \centering       
       \includegraphics[scale=0.4]{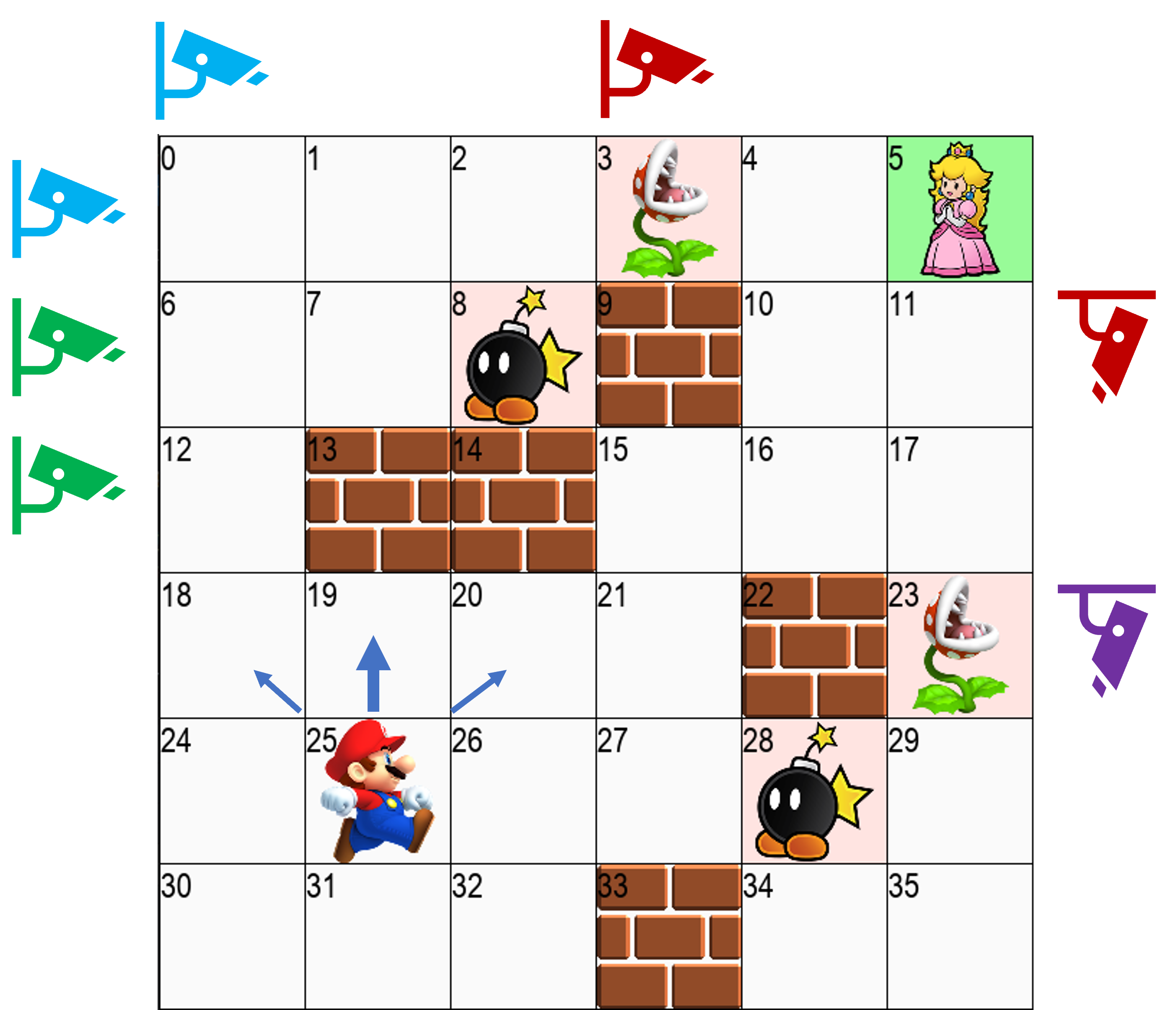}
        \caption{Mario in the stochastic environment. Sensors $\mathsf{S0}$-Green, $\mathsf{S1}$-Red, $\mathsf{S2}$-Blue, $\mathsf{S3}$-Purple.
        }
      \label{fig:mario_bros_example}
\end{figure}

\textbf{Game dynamics and players' information.}
Mario's traversal across the grid unfolds through a selection of actions: $\mathsf{Up}$, $\mathsf{Down}$, $\mathsf{Left}$, and $\mathsf{Right}$. Within this stochastic terrain, each action Mario takes carries a degree of unpredictability. Specifically, there's a probability $p$ that he arrives at the intended cell, and slips to unintended cells with probability $1-p$. For instance, consider the cell $25$: if Mario opts to go $\mathsf{Up}$, he can reach cell $19$ with probability $p$, and with an equal probability of $\frac{(1-p)}{2}$ end up in either of cells $18$ or $20$ (as indicated by the arrows in Fig.~\ref{fig:mario_bros_example}). If any of these potential states is a wall, Mario will only transition to wall-free cells. For example, if Mario is in cell $19$ and selects the action $\mathsf{Up}$ again, he will move exclusively to cell $12$ because the other possible cells, $13$ and $14$, are obstructed by walls.  

 Mario has imperfect observation of his current cell and he depends on the sensor network deployed for localizing himself. The deployed sensor network, however, is vulnerable to adversarial attacks by King Koopa, who aims to prevent Mario from accomplishing his mission. 
 
The sensor network consists of $4$ sensors $\mathsf{S0, S1, S2, S3}$. Each of the sensors represents a deployed Boolean camera sensor (or a range sensor) as shown in Fig.~\ref{fig:mario_bros_example}, which returns $\mathsf{True}$ when Mario is within the range and $\mathsf{False}$ otherwise. In the given scenario in Fig.~\ref{fig:mario_bros_example}, we have that the sensor $\mathsf{S0}$ observes the cells $\{6, 7, 8, 9, 10, 11, 12, 13, 14, 15, 16, 17\}$, while $\mathsf{S1}$ covers the cells $\{6, 7, 8, 9, 10, 11, 3, 15, 21, 27, 33\}$. And $\mathsf{S2}$ monitors the cells $\{0, 1, 2, 3, 4, 5, 6, 12, 18, 24, 30\}$ while $\mathsf{S3}$ monitors the cells $\{18, 19, 20, 21, 22, 23\}$. 

\textbf{Example Setup.} As in the illustrative example, we establish the baseline by conducting the experiment where King Koopa is aware of all sensors in the network and can attack any of them.

We then consider the scenario where Mario has the sensors $\mathsf{S1}$ and $\mathsf{S3}$ hidden. Thus, King Koopa (the adversary) is only aware of the sensors $\mathsf{S0, S2}$ initially, while Mario is aware of the hidden sensors $\mathsf{S1}$ and $\mathsf{S3}$ as well.
We denote the delay in the reaction to the hidden sensor by King Koopa as $\tau$. 

\begin{enumerate}[]
    \item \textbf{Case 1}: \emph{Mario has no hidden sensor.}
    
    In this configuration, we have that P1 can make any of the following sensor queries $\{\{\mathsf{S0, S2}\}, \\ \{\mathsf{S0, S1}\}, \{\mathsf{S1, S2}\}, \{\mathsf{S0, S1, S3}\}  \}$. The adversary can attack the sensor at a time.  
    
    \item \textbf{Case 2}: \emph{Mario has hidden sensors.}
    
    In this setup, we have that Mario can query any of the following sensor queries $\{\{\mathsf{S0, S2}\}, \\ \{\mathsf{S0, S1}\}, \{\mathsf{S1, S2}\}, \{\mathsf{S0, S1, S3}\}  \}$ to determine his position from the sensor network with sensors $\mathsf{S1}$ and $\mathsf{S3}$ hidden. Thus, the adversary initially is unaware of the existence of hidden sensors, sensors $\mathsf{S1}$ and $\mathsf{S3}$, present in the sensor network. The adversary can attack any one of the sensors at a time from the sensors that are known, \ie, one of the sensors $\mathsf{S0,S2}$. And upon revealing \emph{at least one of the two hidden sensors} by Mario, we have that King Koopa \emph{becomes aware of all of the hidden sensors in the network} and can attack any one of all available sensors after a delay in reaction ($\tau$). 

\end{enumerate}

The results of Mario's experimentation for the above cases with varying delay ($\tau$) are summarized in Table~\ref{tab:2}.

\begin{table}[h!]      \vspace{-1ex}

     \caption{Results of the experimentation with Case 1 (without hidden sensors) and Case 2 (with hidden sensors) by varying the delay ($\tau$) for Mario Bros. example.}
    \vspace{1ex}
    \label{tab:2}
 \setlength\tabcolsep{0pt}
        \begin{tabular}{|l|c|c|c|}
        \hline

\textbf{Setup ($\tau$)}
& \textbf{Deceptive Winning Initial States}& \textbf{VoD} \\
\hline
 Case 1 (-) 
&$\mathsf{[4, 10, 11, 16, 17, 5]}$ &$0$ \\ \hline  Case 2 (0) 
&$\mathsf{[4, 10, 11, 16, 17, 5]}$ &$0$ \\ \hline

 Case 2 (1) 
&$\mathsf{15, 20, [4, 10, 11, 16, 17, 5]}$ &$0.105$ \\ \hline
 Case 2 (2) 
&$\mathsf{21, 27, 32, 34, 15, 20, [4, 10, 11, 16, 17, 5]}$ &$0.315$ \\ \hline
 Case 2 (3) 
&$\mathsf{26, 31, 21, 27, 32, 34, 15, 20, [4, 10, 11, 16, 17, 5]}$ &$0.421$ \\ \hline

    \end{tabular}
 \end{table}

\subsubsection{Discussion}

With the results of the Mario Bros example summarized in Table~\ref{tab:2}, like before, we obtain the winning initial states with deception using King Koopa's perceptual game, the delayed attack game, and Mario's initial game. 

We see that the computation time for constructing Mario's initial game reduces with increasing $\tau$. The observation is seen because, in the scenario we have examined, the number of states belonging to the almost-sure winning states in the delayed attack game increases. Consequently, many of these states result in self-loops during the construction of Mario's initial game. Additionally, it's worth noting that the computation time for building the delayed attack game grows with larger values of $\tau$, as anticipated.

In the considered scenario, when $\tau = 1$, Mario uses the hidden sensor $\mathsf{S1}$ with the action $\mathsf{Right}$ while at the cell $15$. This choice is driven by the fact that cells $10$ and $16$ do not have a consistent action to reach the target cell $5$. By utilizing the hidden sensor, Mario is aware of his true current state allowing him to make the correct move. Likewise, the cells $15$ and $16$ also lack a consistent action. Thus, when $\tau=1$, Mario has no winning strategy from the cell $21$. This occurs because even though Mario can use the hidden sensor $\mathsf{S1}$ to determine the true current state if that state happens to be cell $15$, he would once again need to employ the hidden sensor. However, since the hidden sensors are revealed, King Koopa can attack the sensor $\mathsf{S1}$ to ensure that Mario's belief includes both cells $10$ and $16$.

For Case $\tau=0$, the construction times for games $\mathcal{G}^2$, $\mathcal{G}^{+}$, and $\mathcal{G}^0$ are $22.7s$, $75.3s$, and $44.7s$, respectively. The overall time to construct and solve all the games for obtaining the policy is $421.326s$. Similarly, for Case $\tau = 1$, the construction times for these games are $22.7s$, $154.3s$, and $78.1s$, and the overall time is $479.3088s$. In Case $\tau=2$, the construction times are $22.5s$, $234.4s$, and $52.9s$ for $\mathcal{G}^2$, $\mathcal{G}^{+}$, and $\mathcal{G}^0$ respectively, with $804.820s$ as the overall solve time to obtain the policy. Finally, for Case $\tau=3$, the time to obtain the policy is $1191.516s$, and the construction times for the three games are $22.6s$, $482.9s$, and $34.8s$.

\section{Conclusion}
\label{sec:Conclusion}
 
In this study, we studied qualitative planning with joint control and active sensor queries in a stochastic environment, where an adversary can carry out sensor jamming attacks and has perfect observations. Being aware of potential attacks, we investigate whether a class of capability deception---hiding sensors---can gain strategic advantages for the control system. Based on the analysis, we derived that there is no advantage to deviating from the adversary's perceptive rationalizable strategy without revealing hidden sensors. Further, 
almost-sure deceptive winning may be achieved only from states that are positive winning but not almost-sure losing in the original game, in which the attacker has full information. 
  Two case studies validated our algorithm and showed the advantages of deception against sensor attacks.

  In future research, deception could be explored as a countermeasure against more advanced cyber-physical attacks. For example, an attacker may be capable of attacking jointly the sensors and actuators of the cyber-physical system. Defense against joint sensor and actuator attacks may require a stochastic game model and corresponding solutions.  Another potential extension is to study game design that strategically places hidden sensors to induce and capitalize on asymmetric information in the dynamic interaction.  


\appendix

\section{Almost-sure winning algorithm}
\label{appendix:almost_sure_winning_algo}
Given a stochastic reachability game with a partially controllable function, as given in \cite{lcss_paper}, a belief-augmented reachability game with P2's belief and P2's belief of P1's belief is constructed. The belief-augmented game is solved using the Algorithm~\ref{alg:posg-reachability} from \cite{lcss_paper} as given below.

 \begin{algorithm}[h!]
{ 
\caption{Belief-based Almost-Sure Winning Region \cite{lcss_paper}}
\label{alg:posg-reachability}

\begin{algorithmic}[1]
	\item[\textbf{Inputs:}] Two-player reachability game with augmented states $\mathcal{G}$ and  P2's positive winning region $\poswin_2(\mathcal{G})$ in $\mathcal{G}$.
	\item[\textbf{Outputs:}] P1's \ac{asw} region $\asw_1$. 
    \State $j \gets 0$;~$Y_j \gets Q\setminus \poswin_2(\mathcal{G})$
    \Repeat
        \State $k \gets 0$;~$R_k \gets \{q_F\}$
        \Repeat 
            \State $R_{k+1} \gets R_k \cup \prog_1(R_k, Y_j) \cup \prog_2(R_k, Y_j) \cup \prog_N(R_k, Y_j)$
            \State $k \gets k+1$
        \Until{$R_{k+1} = R_k$}
        \State $Y_{j+1} \gets R_k$;~$j \gets j+1$
    \Until{$Y_{j+1} = Y_{j}$}
    \State \Return $\asw_1\gets Y_j$.
 \end{algorithmic}
}      
\end{algorithm}

In the Algorithm ~\ref{alg:posg-reachability} we compute a belief-based almost-sure winning randomized strategy for P1. First, the positive winning region of P2 is obtained using the solution of two-player stochastic games with two-sided perfect observations. These states are the states from which P1 cannot reach the winning state with probability one even when P1 has perfect observations and thus cannot reach the final states when P1 has partial observation. 

Next, the algorithm first initializes a set $Y_0 = Q \setminus \poswin_2(\mathcal{G})$ and it iteratively refines the set $Y_i$ and obtains $Y_{i+1}$. At each iteration $i$, a set of states from which P1 can ensure to stay within $Y_i$ with probability one and reaches the final state with some positive probability in finite steps. The algorithm also uses the functions $\mathsf{Prog}_1(R, Y)$ which outputs a set of states from which P1 has consistent actions to reach $R$ in one step, $\mathsf{Prog}_2(R, Y)$ which outputs a set of states from which P2 cannot prevent reaching $R$ and $\mathsf{Prog}_N(R, Y)$ which outputs a set of states from which the nature state can ensure to remain within $Y$ with probability one and with a positive probability reach the final state.

In the inner loop of the algorithm, it iteratively computes $R_{k+1}$ given $R_k$ until a fixed point is reached. The set $Y$ is then updated with this obtained fixed point and the outer loop is implemented again until an outer fixed point is reached. The thus obtained fixed point is the P1's almost-sure winning region and P1's winning strategy is extracted as given in \cite{lcss_paper}.

\Urlmuskip=0mu plus 1mu\relax
\bibliographystyle{plain}
\bibliography{refs}

\end{document}